\documentclass[11pt]{article}%
\usepackage{tablefootnote}
\usepackage{amsfonts}
\usepackage{mathrsfs}
\usepackage{fancyhdr}
\usepackage{comment}
\usepackage[a4paper, top=2.5cm, bottom=2.5cm, left=2.5cm, right=2.5cm]%
{geometry}
\usepackage{amsmath,amssymb,amsfonts}
\usepackage{graphicx}
\usepackage{float}
\usepackage{hyperref}
\usepackage{natbib}
\usepackage{float}
\usepackage{algorithm}
\usepackage{algorithmic}
\usepackage{xcolor}
\usepackage{enumitem} 
\usepackage{booktabs}
\usepackage{multirow}
\usepackage{stackengine}
\usepackage{verbatim}
\usepackage{caption}
\usepackage{subfigure}

\newtheorem{theorem}{Theorem}

\newtheorem{claim}{Claim}

\newtheorem{corollary}{Corollary}

\newtheorem{lemma}{Lemma}

\newtheorem{remark}{Remark}

\newenvironment{proof}[1][Proof]{\noindent\textbf{#1.} }{\ \rule{0.5em}{0.5em}}
\newcommand{\MMS}{\mathsf{MMS}}

\newcommand{\cV}{\mathcal{V}}

\newcommand{\cU}{\mathcal{U}}

\title{On Hill's Worst-Case Guarantee for Indivisible Bads}

{
	\author{ Bo Li$^1$ \hspace{30pt} Herv\'e Moulin$^2$ \hspace{30pt}  Ankang Sun$^3$ \hspace{30pt} Yu Zhou$^1$\\\small
		$^1$Department of Computing, The Hong Kong Polytechnic University, Hong Kong, China\\\small
		\texttt{comp-bo.li@polyu.edu.hk, csyzhou@comp.polyu.edu.hk}\\ \small
		$^2$Adam Smith Business School, University of Glasgow, Glasgow, UK\\ \small
		\texttt{herve.moulin@glasgow.ac.uk}\\ \small
		$^3$ Warwick Business School, University of Warwick, Warwick, UK\\\small
		\texttt{phd18as@mail.wbs.ac.uk}
}}

\begin{document}

	\maketitle
	
	\begin{abstract}
		When allocating objects among agents with equal rights, people often evaluate the fairness of an allocation rule by comparing their received utilities to a benchmark \textit{share} -- a function only of her own valuation and the number of agents. This share is called a \textit{guarantee} if for any profile of valuations there is an allocation ensuring the share of every agent. When the objects are indivisible goods, Budish [J. Political Econ., 2011] proposed MaxMinShare, i.e., the least utility of a bundle in the best partition of the objects, which is unfortunately not a guarantee. Instead, an earlier pioneering work by Hill [Ann. Probab., 1987] proposed for a share the worst-case MaxMinShare over all valuations with the same largest possible single-object value. Although Hill's share is more conservative than the  MaxMinShare, it is an actual guarantee and its computation is elementary, unlike that of the MaxMinShare which involves solving an NP-hard problem. We apply Hill’s approach to the allocation of indivisible bads (objects with disutilities or costs), and characterise the tight closed form of the worst-case MinMaxShare for a given value of the worst bad. We argue that Hill's share for allocating bads is effective in the sense of being close to the original MinMaxShare value, and there is much to learn about the guarantee an agent can be offered from the disutility of her worst single bad. Furthermore, we prove that the monotonic cover of Hill's share is the best guarantee that can be achieved in Hill's model for all allocation instances. 
	\end{abstract}

	\newpage

	\section{Introduction}
	The task is to fairly allocate a given pile of indivisible objects among 
	agents with equal rights but different preferences. Since
	the very beginning of the fair division literature \citep{Steinhaus49}, allocation rules have been evaluated in part by the worst-case utility they
	guarantee to each participant over all possible utility profiles of the other agents. 
	The higher the guarantee the safer it is for an agent clueless about
	the others' utilities and actions to participate in the allocation process
	defined by the rule.
	
	Formally, the \textit{guarantee} offered by an allocation rule is a mapping from any
	utility function to the corresponding worst-case utility for an agent. 
	This function only depends upon the number of other agents (but not on the
	particulars of the agents) and the domain where their utilities come from.
	When the objects are divisible and desirable (i.e., \textit{goods}), and utilities are
	additive, the optimal (largest feasible) guarantee is $\frac{1}{n}\cdot v_{i}(M)$,
	where $M$ is the set of goods, $n$ is the number of agents, and $i$ is a generic agent with utility function $v_i$ \cite{dubins1961cut}. 
	But in all 
	the important practical contexts where the objects are indivisible while utilities
	remain additive, the search for maximal guarantees (those that cannot be
	improved over the entire domain of utilities) cannot be that simple. The
	difficulty is obvious when we consider the ``one diamond and several
	worthless rocks'' example: unless we throw away the diamond,
	all agents but
	one end up with a negligible fraction of $v_{i}(M)$.

	To capture exactly the diamond effect when indivisible goods are distributed,
	the concept of MaxMinShare \citep{budish2011combinatorial} has been intensely studied over the last decade 
	\citep{DBLP:journals/talg/AmanatidisMNS17,DBLP:journals/jacm/KurokawaPW18,DBLP:conf/sigecom/HuangL21}. 
	An agent's MaxMinShare is motivated by an imaginary divide-and-choose experiment: the agent gets the chance to partition the objects into $n$ bundles, but is the last one to choose one bundle. 
	Then, the agent's MaxMinShare is the
	utility of her worst share in the best $n$-partition of the objects.
	MaxMinShare bears some disadvantages.
	On the one hand, the definition is not trivial and computing its value involves solving an NP-hard problem.
	On the other hand, 
	in some rare cases, the MaxMinShare is not a feasible
	guarantee \citep{DBLP:conf/sigecom/ProcacciaW14};
	so far the best-known approximation is that a $(\frac{3}{4}+o(1))$ fraction of MaxMinShare can
	be guaranteed and implemented in polynomial time \citep{DBLP:conf/sigecom/GhodsiHSSY18,DBLP:journals/ai/GargT21}. 
	
	Back to 1980s, \citet{hill1987partitioning} also investigated how the indivisibility of the objects affect the agent's guaranteed share {by restricting} attention to additive utility functions $v$ such that $v(M)=1$
	(without loss of generality) and the most valuable object of $v$ is worth {$\alpha$, $0 < \alpha < 1$}; we
	write $\mathcal{V}(\alpha )$ for this subdomain of additive valuations.
	Hill proposed to study the worst-case MaxMinShare among all valuations in $\mathcal{V}(\alpha)$, which is referred to as the {\em Hill's
		share} throughout this paper. 
	In \citep{hill1987partitioning}, Hill computed for every $n\geq 2$ a function $V_{n}:[0,1] \to [0,\frac{1}{n}]$, which lower-bounds Hill's share.
	By definition, $V_{n}(\alpha )$ is also a lower bound on the MaxMinShare of every utility in $\mathcal{V}(\alpha)$.
	Depending on $\alpha 
	$ the guarantee $V_{n}(\alpha )$ may or may not improve upon the $%
	\frac{3}{4}$-approximate MaxMinShare guarantee,
	but its great advantage is that whether a given allocation
	meets the guarantee for a given utility is immediately verifiable.
	Furthermore, Hill proved that if every
	agent's utility is in $\mathcal{V}(\alpha )$, it is always possible to
	simultaneously give each agent a share worth at least $V_{n}(\alpha )$, i.e., $V_{n}(\cdot)$ is a guarantee.
	\citet{DBLP:conf/wine/MarkakisP11} proved a stronger result: {the share $V_n(\alpha_i)$ where $\max_{e\in M}v_{i}(e) = \alpha _{i}$}
	is a bona fide guarantee over the full domain of additive and
	nonnegative utilities. Moreover, an allocation implementing these individual
	guarantees can be computed in polynomial time.
	\citet{DBLP:journals/tcs/GourvesMT15} found that $V_{n}(\alpha )$ is not the tight characterisation of Hill's share and proved a tighter function.
	An interesting fact is that the tight function is not monotone in $\alpha$, but its exact computation is still open.
	
	All the aforementioned work, as well as the majority of fair division literature, focuses on the allocation of goods, and the mirror problem of bads (undesirable 
	objects like chores, liabilities when a partnership is dissolved, etc.; see \cite{DBLP:journals/mp/LenstraST90}) is not as well understood as that of goods, which motivates the current work.

	\subsection{Our Problem and Results}

	We apply Hill's approach to the allocation of indivisible bads 
	and prove a set of results parallel to those just
	mentioned. 
	The diamond effect now becomes the ``chore from hell'' effect where the 
	\textit{dis}utility is concentrated in a single bad, and now $\mathcal{V}(\alpha )$ collects all disutility functions where the value of the
	worst bad equals $\alpha $, maintaining the normalisation $v(M)=1$.
	
	Our results for bads resemble those just mentioned for goods, and in
	addition, they make the connection between Hill's share and MinMaxShare
	({\textit{the largest disutility of a share in the best partition of the bads}}).
	To be more precise, we compute first the tight characterisation of Hill's share, refined to problems with a given number $m$ of bads, i.e., the exact upper bound $%
	\Delta _{n}^{\oplus }(\alpha ;m)$ of the MinMaxShare in the domain $\mathcal{%
		V}(\alpha; m)$, where ${\cV}(\alpha; m)$ contains the valuations over $m$ objects with the highest disutility being $\alpha$.
	This result is stated in Theorem \ref{thm:homogeneous}.
	If $m$ is not restricted, i.e., ${\cV}(\alpha) = \bigcup_{m} {\cV}(\alpha; m)$ and $\Delta _{n}^{\oplus }(\alpha) = \max_m \Delta _{n}^{\oplus }(\alpha ;m)$, we illustrate the function $\Delta _{n}^{\oplus }(\alpha)$  for $n=2,3$ in Fig. \ref{fig:result1:homo}. 
	Just like \citet{DBLP:journals/tcs/GourvesMT15} observed for the problem of goods, this function is not monotone in $\alpha $.
	In passing, we tighten the bounds proposed by \citet{hill1987partitioning} 
	and \citet{DBLP:journals/tcs/GourvesMT15} for the
	worst-case MaxMinShare in the two-agent problem of goods; see Remark \ref{rem:n=2}.

	\begin{figure}[htbp]
		\centering
		\begin{minipage}[t]{0.48\textwidth}
			\centering
			\includegraphics[width=7.5cm]{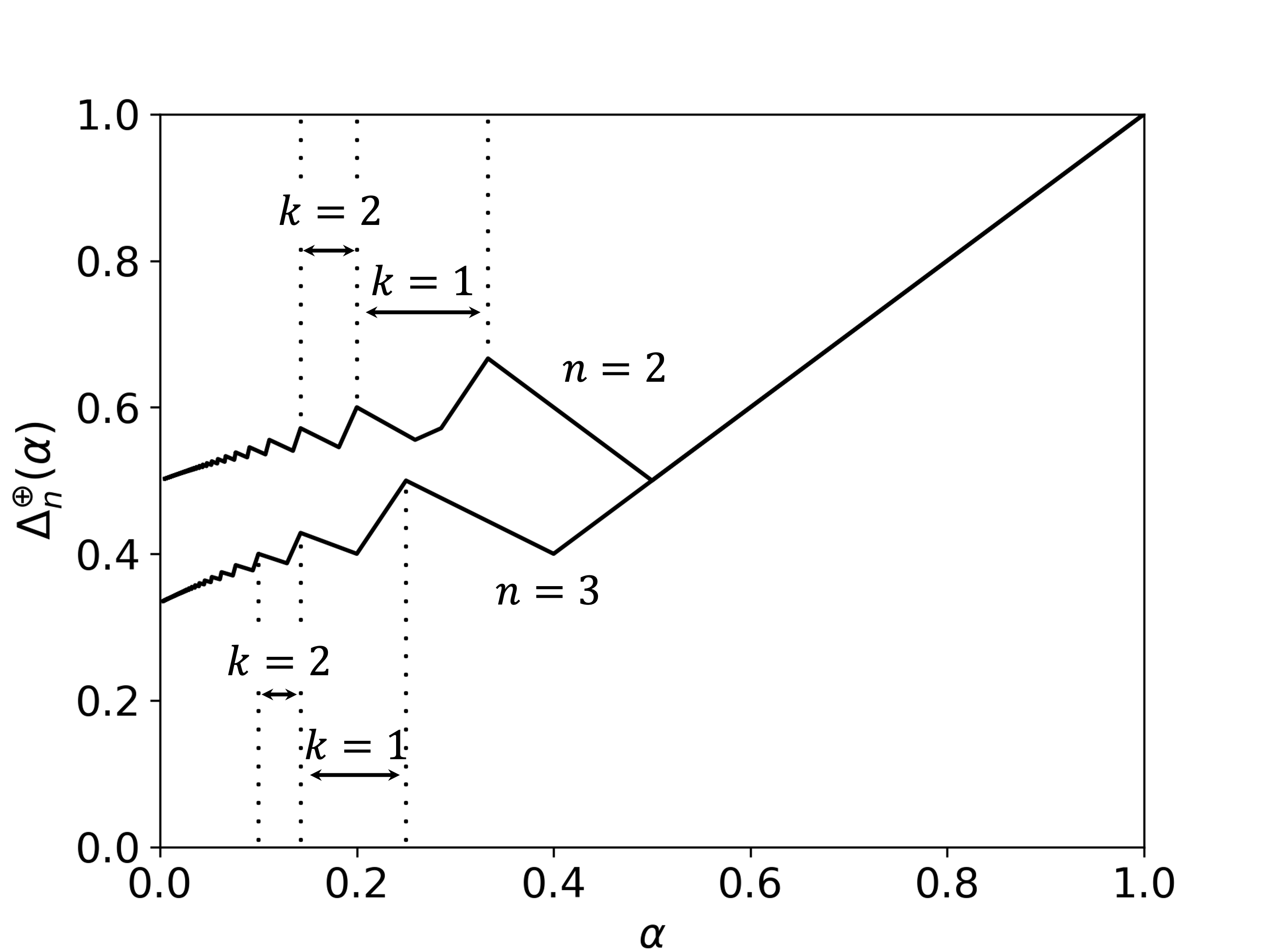}
			\caption{Hill's share $\Delta _{n}^{\oplus }(\alpha)$ when $n=2$\\
				and $3$ and $m$ is not restricted.}
			\label{fig:result1:homo}
		\end{minipage}
		\begin{minipage}[t]{0.48\textwidth}
			\centering
			\includegraphics[width=7.5cm]{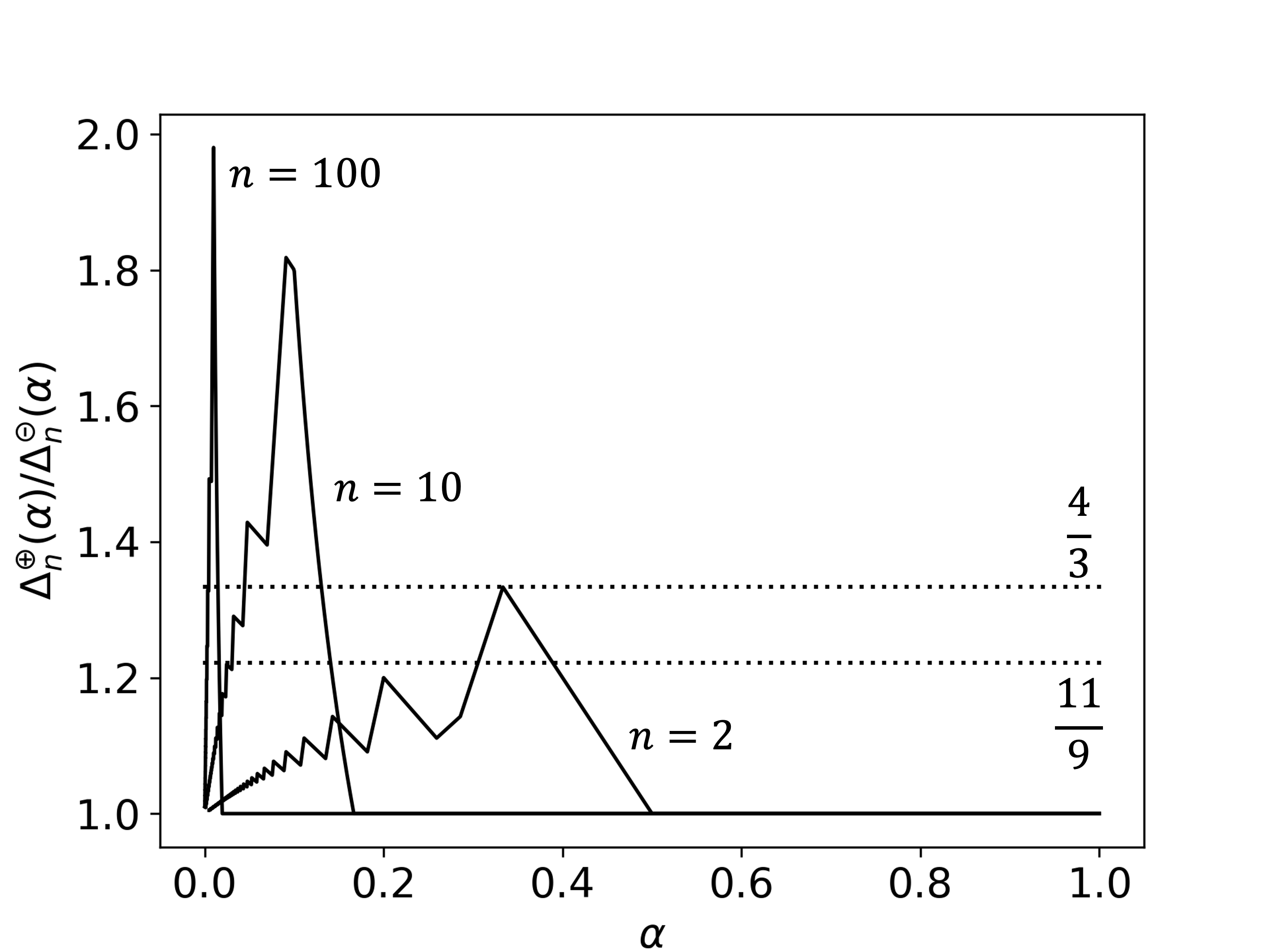}
			\caption{The ratio between the upper and lower bounds of the MinMaxShare of valuations in $\mathcal{V}(\alpha)$.
				$4/3$ and $11/9$ are two fractions of the MinMaxShare known to be achievable. }
			\label{fig:ratio:theoretical}
		\end{minipage}
	\end{figure}

	Compared to the MinMaxShare, Hill's share $\Delta _{n}^{\oplus } (\alpha;m)$ is immediately verifiable, whereas deciding whether (a multiple of) the MinMaxShare is met at a given allocation involves solving an NP-hard problem.
	Moreover, {the function $\alpha \to \Delta_n^{\oplus}(\alpha; m)$ relating the guaranteed share to the disutility of the worst bad (relative to total disutility) is a transparent hard design constraint of which all participants should be aware.}
	Although $\Delta _{n}^{\oplus }(\alpha;m)$ seems more conservative than the MinMaxShare of a specific disutility function, we argue that $\Delta _{n}^{\oplus }(\alpha;m)$ is approximately as effective as MinMaxShare.
	First, $\Delta _{n}^{\oplus }(\alpha;m)$ is at most twice the MinMaxShare of every disutility in ${\cV}(\alpha;m)$.
	We plot the exact ratio of $\Delta _{n}^{\oplus }(\alpha)$ and the best MinMaxShare of disutilities in $\cV(\alpha)$ for every $\alpha$ in Fig. \ref{fig:ratio:theoretical} when $n=2, 10$ and $100$. 
	As we can see, although the largest ratio may reach 2 (only happens when $n$ is large), for most values of $\alpha$, the ratio is not far from 1.
	In particular, $\Delta _{n}^{\oplus }(\alpha)$ outperforms the fractions of the MinMaxShare known to be
	implementable ($\frac{4}{3}$ by \citet{DBLP:journals/teco/BarmanK20} and $%
	\frac{11}{9}$ by \citet{DBLP:conf/sigecom/HuangL21}) for most $\alpha$ no matter what values $n$ has.
	Besides the above worst-case comparison, in Section \ref{sec:experiments}, we conduct numerical experiments with synthetic and real-world data to illustrate the real distances between  Hill’s share and MinMaxShare.
	The experiments show that Hill's share is actually very close  to (e.g., within 1.1 fraction of) the MinMaxShare for the majority of the  instances.

	Finally, we obtain the main result of this work -- a counterpart for bads of Hill's guarantee for goods improved by \citet{DBLP:conf/wine/MarkakisP11}. Letting $V_{n}(\alpha ;m)$ denote the monotonic cover of $\Delta_{n}^{\oplus }(\alpha; m)$ with respect to $\alpha$, Theorem \ref{thm:heterogeneous} shows that the share $V_n(\alpha_i; m)$ is a guarantee over the full domain of additive disutilities with $m$ bads. We also provide an algorithm to implement this guarantee in polynomial time. To the best of our knowledge no other similarly simple guarantee for allocating bads has been identified.

	\subsection{More Relevant Literature}
	\label{sec:relatedworks}

	The properties known as proportionality up to one object (Prop1) and up
	to any object (PropX) offer different relaxations of the equal share $%
	\frac{1}{n}v_{i}(M)$ when the objects (goods or bads) are indivisible  \cite{moulin2019fair, DBLP:journals/orl/AzizMS20}. 
	These relaxations require the equal share to be satisfiable if at most one object is added or removed. 
	Like the Hill's
	guarantees for goods or bads we discuss here, they are immediately
	verifiable, but unlike these they are not always preserved by Pareto
	improvements, a serious limitation of their implementation.
	They also do not provide agents with any guaranteed utility or disutility. 
	The same remark applies to the popular ex-post tests no envy up to one
	object, or up to any object \cite{DBLP:conf/sigecom/LiptonMMS04,budish2011combinatorial,DBLP:journals/teco/CaragiannisKMPS19}.
	Another easily verifiable test is the truncated proportional share (ATP) bound of \cite{DBLP:conf/wine/BabaioffEF22},
	but unlike Hill's guarantees, it improves upon the MaxMinShare for goods so it is not a feasible guarantee.
	
	Intuitively, the allocation of undesirable bads is the mirror image of that of goods. 
	However, adapting the results is not a simple matter of switching signs.
	{For instance when objects are indivisible the approximations Prop1 and PropX behave quite differently for goods and bads \cite{moulin2019fair, DBLP:journals/orl/AzizMS20}.\footnote{%
			See also \cite{DBLP:journals/scw/BogomolnaiaMSY19} for the competitive equilibrium from
			equal incomes when objects are divisible.} Our results confirm this general observation: in our case the general allure of
		the critical functions $\Delta _{n}^{\oplus}$ is the same for goods and for bads, but the details and the proofs are quite different.}
	We refer the readers to recent surveys by \cite{moulin2019fair} and \cite{DBLP:journals/sigecom/Aziz22} for a more detailed discussion 
	on the fair division of indivisible goods and bads.
	In particular, \cite{DBLP:journals/sigecom/Aziz22} explicitly listed computing Hill's guarantee for bads as an open problem.

	\section{Preliminaries}
	\label{sec:pre}

	For any positive integer $k$, let $[k]=\{1,\ldots,k\}$.
	We consider allocating $m$ indivisible objects, denoted by $M=[m]$, among $n$ agents,
	and let $\mathcal{A}dd(M)$ be the domain made of the nonnegative additive disutility functions $v$ on object set $M$, normalised without loss of generality, as follows%
	\[
	v(S)=\sum_{e\in S}v(\{e\})\mbox{ for all } S\subseteq M \mbox{ and }v(M)=1.
	\]
	Following the convention of the literature, disutility functions are also called valuations.
	For simplicity, we write $v(e)$ to represent $v(\{e\})$ for each $e\in M$.
	For any $\alpha \in \lbrack 0,1]$, the subdomain $\mathcal{V}(\alpha ;m) \subseteq \mathcal{A}dd(M)$ is defined by the property $\max_{e\in M}v(e)=\alpha $ and 
	$\mathcal{U}(\alpha ;m)$ by $v(e)\leq \alpha $ for all $e\in M$.
	According to the definitions, $\cV (\alpha;m)\subseteq \cU(\alpha;m)$
	for any valid pair of $\alpha$ and $m$.
	Note that, since the functions are all normalised, $\mathcal{V}(\alpha ;m)$ is only well
	defined if $\alpha \times m\geq 1$, equivalently for $m\geq m_{\ast }=\lceil 
	\frac{1}{\alpha }\rceil $ (the upper integer part of $\frac{1}{\alpha }$).
	
	An allocation, denoted by $ \mathbf{A} = (A_1,\ldots,A_n)$, is a partition of $M$ into $n$ disjoint subsets of objects; 
	note that some of these subsets can be empty.
	The set of all allocations is denoted by $\mathcal{X}_{n}(M)$.
	The
	MinMaxShare 
	(MMS),
	when there are $n$ agents,
	of the disutility $v\in \mathcal{A}dd(M)$ is defined as%
	\[
	\MMS_{n}(v)=\min_{ \mathbf{A} \in \mathcal{X}_{n}(M)}\max_{1\leq \ell \leq n}v(A_{\ell }).
	\]
	We next define the upper and lower bounds of MinMaxShare among all disutilities in $\mathcal{V}(\alpha;m)$,
	\begin{align*}
		\Delta _{n}^{\oplus }(\alpha ;m)&=\max_{v\in \mathcal{V}(\alpha
			;m)}\MMS_{n}(v); \mbox{ and }\\ 
		\Delta _{n}^{\circleddash }(\alpha ;m)&=\min_{v\in 
			\mathcal{V}(\alpha ;m)}\MMS_{n}(v).  
	\end{align*}
	The upper bound $\Delta _{n}^{\oplus }(\alpha ;m)$ (i.e., the worst-case MinMaxShare) is called {\em Hill's share}, and we use these terms interchangeably in this paper.

	It is not difficult to obtain the below formula of $\Delta_n^{\circleddash}(\alpha; m)$, whose formal proof is in the appendix. 
	\begin{lemma}
		\label{lem:bestMMS_m}
		Given $0 < \alpha < 1$, $n \ge 2$, and $m \ge \lceil \frac{1}{\alpha} \rceil$, $\Delta_{n}^{\circleddash}(\alpha; m)$ is as follows: 
		\begin{align*}
			\Delta _{n}^{\circleddash}(\alpha; m) = 
			\left\{\begin{array}{ll} 
				\alpha, &\text{ if } \alpha > \frac{1}{n}, \\
				\frac{1}{n}, &\text{ if } \alpha = \frac{1}{kn}, \text{ or } \frac{1}{(k+1)n} < \alpha < \frac{1}{kn} \text{ and } m \ge kn + n\\
				k\alpha + \frac{1-kn\alpha}{m-kn}, &\text{ if } \frac{1}{(k+1)n} < \alpha < \frac{1}{kn} \text{ and } m \le kn + n - 1
			\end{array}\right.
		\end{align*}
		for some integer $k \ge 1$. 
	\end{lemma}
	
	Computing $\Delta _{n}^{\oplus }(\alpha ;m)$ is non-trivial, as shown in Section \ref{sec::upper-bound-minmaxshare}, 
	but the following lemma, proved in the appendix, presents two simple properties. 
	\begin{lemma}
		\label{lem:worstMMS_property}
		(1) $\Delta _{n}^{\oplus }(\alpha ;m)$ is weakly decreasing in $n$;
		(2) $\Delta _{n}^{\oplus }(\alpha ;m)$ is weakly increasing in $m$ from $\lceil \frac{1}{\alpha} \rceil $ to $\lceil \frac{2}{\alpha} \rceil - 1$ and constant thereafter. 
	\end{lemma}
	
	By the second property in Lemma \ref{lem:worstMMS_property}, and also following \cite{hill1987partitioning,DBLP:conf/wine/MarkakisP11,DBLP:journals/tcs/GourvesMT15}, we also consider the case when $m$ is not restricted, or equivalently, $m = \infty$.
	Let $\mathcal{V}(\alpha) = \bigcup_{m} \mathcal{V}(\alpha,m)$ and $\mathcal{U}(\alpha) = \bigcup_{m}\mathcal{U}(\alpha ;m)$.
	Accordingly, we have 
		
	\begin{align}
		\Delta _{n}^{\oplus }(\alpha)&=\max_{v\in \mathcal{V}(\alpha)}\MMS_{n}(v); \mbox{ and } \nonumber  \\ 
		\Delta _{n}^{\circleddash }(\alpha)&=\min_{v\in 
			\mathcal{V}(\alpha)}\MMS_{n}(v).  \nonumber
	\end{align}
	By Lemma \ref{lem:bestMMS_m}, $\Delta _{n}^{\circleddash }(\alpha) = \max\{\alpha, 1/n\}$.

	Hill's share $\Delta_n^{\oplus}(\alpha; m)$ (and $\Delta_n^{\oplus}(\alpha)$) behave much like the MinMaxShare in the following senses. 
	First, for any $v\in \mathcal{V}(\alpha; m)$ there is an allocation $(A_1, \ldots, A_n)$ such that $v(A_i) \le \Delta_n^{\oplus}(\alpha; m)$ for all $i$. 
	This follows from the definition of the MinMaxShare plus that $\Delta_n^{\oplus}(\alpha; m)$ is an upper bound of the MinMaxShare. 
	Second, the \textit{max} in the definition of $\Delta_n^{\oplus}(\alpha; m)$ is achieved by some $v^* \in \mathcal{V}(\alpha; m)$; that is, $\Delta_n^{\oplus}(\alpha; m) = \MMS_n(v^*)$. 
	This is because $\mathcal{V}(\alpha; m)$ is a compact set and all the functions are continuous. 
	Then we know that for any allocation $(B_1,\ldots,B_n)$ there is some $i$ such that $v^*(B_i) \ge \Delta_n^{\oplus}(\alpha; m)$. 
	Note that these two facts have nothing to do with what the function $\Delta_n^{\oplus}(\alpha; m)$ actually looks like and they can be easily adapted to $\Delta_n^{\oplus}(\alpha)$.

	\section{Characterising Hill's Share}\label{sec::upper-bound-minmaxshare}

	\subsection{Main Result}
	
	We now characterise Hill's share, i.e., the exact upper bound of the MinMaxShare values, $\Delta _{n}^{\oplus }(\alpha ;m)$ and  $\Delta _{n}^{\oplus }(\alpha)$.
	For any integers $n\ge 2$ and $k\ge 0$, define the following real intervals:
	
	$$
	\begin{aligned}
		D(n,k) &= \left( \frac{1}{kn+n+1}, \frac{k+2}{n(k+1 )^2  + k + 2 } \right]\\
		I(n,k) &= 
		\left( \frac{k+2}{n(k+1 )^2  + k + 2 }, \frac{1}{kn+1} \right].
	\end{aligned}
	$$
	It is not hard to check that all the intervals are well-defined, non-overlapping, and $\bigcup_{k\ge 0} (D(n,k) \cup I(n,k)) = (0,1]$.

	Our first main theorem gives the tight characterisation of Hill's share. 

	\begin{theorem}
		\label{thm:homogeneous}
		For any $0< \alpha <1$, $n\ge 2$, and $m\ge \lceil\frac{1}{\alpha} \rceil$,  
		\begin{align}
			\Delta_n^{\oplus} ( \alpha; m ) = 
			\left\{\begin{array}{ll} 
				\frac{k+2}{k+1}\cdot\frac{1-\alpha}{n}, & \text { if } \alpha \in D(n,k) \text { and }  m \geq kn + n + 1,\\ 
				( k + 1 )\alpha, & \text { if } \alpha \in D(n,k)
				\text { and }  m \leq kn + n, \\
				(k+1)\alpha, & \text { if } 
				\alpha \in I(n,k)
			\end{array}\right.
		\end{align}
		for any integers $ n \geq 2$ and $ k \geq 0$ except $ n=2$ and simultaneously $ k = 1$.
		If $ n = 2$ and $ k = 1$,
		$\Delta _ 2^{\oplus} ( \frac{1}{3}; 3) = \frac{2}{3}$, $\Delta_ 2^{\oplus} ( \alpha; 4) = 2 \alpha$ for $\alpha \in [\frac{1}{4}, \frac{1}{3}]$, and $\Delta _ 2^{\oplus} ( \alpha; 5)$ is as follows:
		\begin{align}
			\Delta _ 2^{\oplus} ( \alpha; 5)  = 
			\left\{\begin{array}{ll} 
				\frac{3 - 3\alpha}{ 4 }, & \text { if } \alpha \in ( \frac{1}{5}, \frac{3}{11} ],\\ 
				2\alpha, & \text { if } \alpha \in ( \frac{3}{11}, \frac{1}{3} ], \\
			\end{array}\right.
		\end{align}
		and for $ m \geq 6 $,
		
		\begin{align}
			\Delta_ 2^{\oplus} ( \alpha; m) = 
			\left\{\begin{array}{ll} 
				\frac{3 - 3\alpha}{ 4 }, & \text { if } \alpha \in (\frac{1}{5}, \frac{7}{27}] \\ 
				\alpha + \frac{2 - 2\alpha}{5}, & \text { if } 
				\alpha \in (\frac{7}{27}, \frac{2}{7}] \\
				2\alpha, & \text { if } \alpha \in ( \frac{2}{7}, \frac{1}{3}].
			\end{array}\right.
		\end{align}
	\end{theorem}
	
	Theorem \ref{thm:homogeneous} directly implies the result when the number of objects is not restricted, as shown in the following corollary.
	
	\begin{corollary}
		\label{coro:homogeneous}
		For any $0< \alpha <1$, $n\ge 2$,  
		$\Delta _{n}^{\oplus}(\alpha) = \max\limits_{m \geq \lceil \frac{1}{\alpha} \rceil} \Delta_n^{\oplus}(\alpha; m)$. 
	\end{corollary}

	Actually, Corollary \ref{coro:homogeneous} is a special case of Theorem \ref{thm:homogeneous} when $m$ is sufficiently large (e.g., $m \ge \lceil \frac{2}{\alpha} \rceil - 1$ by Lemma \ref{lem:worstMMS_property}).
	Recall we illustrated $\Delta_2^{\oplus}(\alpha)$ and $\Delta_3^{\oplus}(\alpha)$ in Fig. \ref{fig:result1:homo}.
	We observe two interesting and somewhat unintuitive facts about Theorem \ref{thm:homogeneous}.
	First, $\Delta_n^{\oplus}(\cdot)$ is not monotone in $\alpha$, just like \citet{DBLP:journals/tcs/GourvesMT15} observed for the problem with goods.
	To characterise $\Delta _{n}^{\oplus }(\alpha ;m)$, we want to understand the worst-case disutility in $\cV(\alpha;m)$, for which the objects can be hardly partitioned into bundles with similar disutilities.  
	Intuitively, when the single-object disutility gets larger, it becomes harder to find such a balanced partition.
	However, this turns out to be imprecise.
	Second, the case of $n=2$ makes a difference from $n\ge 3$. When $n=2$ and $k=1$, there are three steps in $\Delta _{n}^{\oplus }(\cdot)$: the worst-case MinMaxShare has two increasing intervals with different slops following a decreasing interval.
	For all the other values of $n$ and $k$, there are two intervals with one decreasing and the other increasing.

	\begin{remark}
		\label{rem:n=2}
		When $n=2$ the problem of bads and that of goods are the same, since maximising the minimum bundle by partitioning the objects into two bundles is equivalent to minimising the maximum bundle. 
		For $n=2$, \citet{DBLP:journals/tcs/GourvesMT15} provided a lower bound of the MaxMinShare for goods which is not tight.
		It can be verified that $1-\Delta_2^{\oplus}(\alpha)$ is strictly larger than their bound when $\alpha \in (\frac{1}{5}, \frac{3}{10})$ (Definition 2 in \citep{DBLP:journals/tcs/GourvesMT15}).
		Thus, as a byproduct,  Corollary \ref{coro:homogeneous} improves the result in \cite{DBLP:journals/tcs/GourvesMT15} for goods with $n=2$ by giving the tight worst-case bound, i.e.,
		\[
		\min_{v\in \mathcal{V}(\alpha)}\max_{ \mathbf{A} \in \mathcal{X}_{2}(M)}\min_{1\leq \ell \leq 2}v(A_{\ell }) = 1- \Delta_2^{\oplus}(\alpha).
		\]
		In Remark \ref{remark:heter}, we show how to extend this result to two non-identical disutilities. 
	\end{remark}

	\subsection{Roadmap for the Proof of Theorem~\ref{thm:homogeneous}}
	\label{sec:homo:proof:road}
	
	As we have discussed, after $m$ reaches a certain value (e.g., $m \ge \lceil \frac{2}{\alpha} \rceil - 1$ by Lemma \ref{lem:worstMMS_property}), Hill's share does not increase anymore, and thus Corollary \ref{coro:homogeneous} is a special case of Theorem \ref{thm:homogeneous} when $m$ is sufficiently large.
	Therefore, in this subsection, we first prove Corollary \ref{coro:homogeneous}, and
	in the appendix, we carefully
	discuss Hill's share when $m$ is not sufficiently large, which will complete the proof of Theorem~\ref{thm:homogeneous} accordingly. 
	Further, we also defer the proof of case $n=2$ and $k=1$ to the appendix, which makes a difference from the other cases and requires a more involved analysis.

	We prove Corollary \ref{coro:homogeneous} by contradiction,
	and assume that there exists a disutility $v \in \mathcal{V(\alpha)}$ whose MinMaxShare is larger than $\Delta_n^{\oplus}(\alpha)$. 
	Let $\mathbf{A} = (A_1, \ldots, A_n)$ be a lexicographical MinMax allocation of $v$;
		that is, the largest disutility of bundles in $\mathbf{A}$ is the minimised over all allocations, and among these allocations the second largest disutility is minimised, and so on. 
		Without loss of generality, assume  $v(A_1)\ge\cdots\ge v(A_n)$ and $v(A_1) = \MMS_n(v) > \Delta_n^{\oplus}(\alpha)$. 
		Let $E_\alpha$ denote the subset of objects whose disutilities are exactly $\alpha$, i.e., $E_\alpha = \{e\in M \mid v(e)=\alpha\}$. 
		It can be verified that $\Delta_n^{\oplus}(\alpha) \ge (k+1)\alpha$ (this is also illustrated in Fig. \ref{fig:result1:homo}), which gives $v(A_1) > (k+1)\alpha$. 
		Moreover, since $v(e) \le \alpha$ for any $e \in M$, $|A_1| \ge k+2$. 
		We have the following property.
		
		\begin{claim}\label{clm:diff}
			Letting $j$ be an agent in $N\setminus\{1\}$, for any $S_1\subseteq A_1$ and $S_j\subseteq A_j$ such that $v(S_1) > v(S_j)$, $v(S_1) - v(S_j) \ge v(A_1) - v(A_j)$.
		\end{claim}
		\begin{proof}
			For the sake of contradiction, we assume that there exist $S'_1\subseteq A_1$ and $S'_{j'}\subseteq A_{j'}$ such that $v(S'_1) > v(S'_{j'})$ and $v(S'_1) - v(S'_{j'}) < v(A_1) - v(A_{j'})$.
			Then we construct another allocation $\mathbf{B} = (B_1, \ldots, B_n)$ by exchanging $S'_1$ and $S'_{j'}$, 
			i.e., $B_1 = A_1 \setminus S'_1 \cup S'_{j'}$, $B_{j'} = A_{j'} \setminus S'_{j'} \cup S'_1$ and $B_j = A_j$ for any $j \in N \setminus \{1, {j'}\}$. 
			It follows that $v(B_1) < v(A_1)$, $v(B_{j'}) < v(A_1)$ and $v(B_j) = v(A_j)$ for any $j \in N \setminus \{1, {j'}\}$, which contradicts the assumption that $\mathbf{A}$ is a lexicographical MinMax allocation of $v$.
		\end{proof}
		
		The contraposition of Claim \ref{clm:diff} gives the following. 
		\begin{claim}\label{clm:v_Aj}
			Letting $j$ be an agent in $N\setminus\{1\}$,  for any $S_1\subseteq A_1$ and $S_j\subseteq A_j$ such that $v(A_j \setminus S_j \cup S_1) < v(A_1)$, $v(S_j) \ge v(S_1)$.
		\end{claim}

		As a warm-up,
		we start from the case with 
		large $\alpha$, where $k = 0$, and distinguish two subcases depending on the domain of $\alpha$.
		
		
		
		\medskip
		
		\noindent{\bf Case 1: $n\ge 2$ and $k=0$}
		
		\medskip
		
		{\bf Subcase 1.1: $\alpha\in D(n,0)$}
		
		\medskip
		
		
		When $\alpha\in D(n,0)$, $\frac{1}{n+1}< \alpha\le \frac{2}{n+2}$ and $v(A_1) > \Delta_n^{\oplus}(\alpha)= \frac{2-2\alpha}{n}$.
		If $E_\alpha \cap A_1 \neq \emptyset$, there exists $e^* \in A_1$ such that $v(e^*) = \alpha < v(A_1)$. 
		Then Claim \ref{clm:v_Aj} gives a lower bound of $v(A_j)$ for any $j\in N\setminus\{1\}$, i.e., $v(A_j) \ge v(e^*) = \alpha$. 
		Summing up these lower bounds leads to the following contradiction
		\[
		1 = \sum_{j\in N}v(A_j) > \frac{2-2\alpha}{n} + (n-1)\cdot\alpha  = \frac{(n+1)(n-2)\alpha + 2}{n} \ge 1,
		\]
		where the last inequality is because $\alpha > \frac{1}{n+1}$ and $n \ge 2$. 
		
		Therefore, $E_\alpha \cap A_1 = \emptyset$. 
		Then by the definition of $\mathcal{V}(\alpha)$, there must exist $j' \in N\setminus\{1\}$ such that $E_\alpha \cap A_{j'} \neq \emptyset$, and thus $v(A_{j'}) \ge \alpha$. 
		Recall that $|A_1| \ge k+2 = 2$, this implies there exists $S \subseteq A_1$ such that $v(A_1) > v(S) \ge \frac{1}{2}v(A_1) > \frac{1-\alpha}{n}$. 
		According to Claim \ref{clm:v_Aj}, $v(A_j) \ge v(S) > \frac{1-\alpha}{n}$ holds for any $j \in N \setminus \{1, j'\}$. 
		As a result, 
		\[
		1 = \sum_{j\in N}v(A_j) > \frac{2-2\alpha}{n} + \alpha + (n-2)\cdot\frac{1-\alpha}{n} = 1,
		\]
		which is also a contradiction. Therefore, $v(A_1) > \Delta_n^{\oplus}(\alpha)$ never holds when $\alpha\in D(n,0)$. 
		
		For the other direction, the disutility function for this subcase (see Table \ref{tab:k=0-D}) contains one object with disutility $\alpha$ and $n$ objects with disutility $\frac{1-\alpha}{n}$. 
		Since $\frac{1}{n+1} < \alpha \le \frac{2}{n+2}$, it follows that $\frac{1-\alpha}{n} < \alpha \le 2\cdot \frac{1-\alpha}{n}$. 
		Clearly, the MinMaxShare of this disutility function is $2\cdot \frac{1-\alpha}{n} = \Delta_n^{\oplus}(\alpha)$. 
		
		\begin{table}[H]
			\centering
			\begin{tabular}{c|c}
				\hline
				Object Disutility & Quantity \\
				\hline
				$\alpha$ & 1 \\
				$\frac{1-\alpha}{n}$ & $n$ \\
				\hline
			\end{tabular}
			\caption{Disutility function for Subcases 1.1 and 1.2.}
			\label{tab:k=0-D}
		\end{table}
		
		{\bf Subcase 1.2:  $\alpha\in I(n,0)$}
		
		\medskip
		
		When $\alpha\in I(n,0)$, by similar reasonings, we can show that $v(A_1) > \Delta_n^{\oplus}(\alpha)$ does not hold, either. 
		In this subcase, $\frac{2}{n+2}< \alpha\le 1$ and $\Delta_n^{\oplus}(\alpha)= \alpha$.
		If $E_\alpha \cap A_1 \neq \emptyset$, there exists $e^* \in A_1$ such that $v(e^*) = \alpha < v(A_1)$ and Claim \ref{clm:v_Aj} gives a lower bound of $v(A_j)$ for any $j\in N\setminus\{1\}$, i.e., $v(A_j) \ge v(e^*) = \alpha$. 
		Summing up these lower bounds leads to the following contradiction
		\[
		1 = \sum_{j\in N}v(A_j) > n\alpha > \frac{2n}{n+2} \ge 1, 
		\]
		where the last inequality is because $n \ge 2$. 
		
		Therefore, it must hold that $E_\alpha \cap A_1 = \emptyset$ and moreover, there exists $j' \in N\setminus\{1\}$ with $E_\alpha \cap A_{j'} \neq \emptyset$.
		Thus, $v(A_{j'}) \ge \alpha$. 
		Since $|A_1| \ge k+2 = 2$, there exists $S \subseteq A_1$ such that $v(A_1) > v(S) \ge \frac{1}{2}v(A_1) > \frac{\alpha}{2}$. 
		According to Claim \ref{clm:v_Aj}, $v(A_j) \ge v(S) > \frac{\alpha}{2}$ holds for any $j \in N \setminus \{1, j'\}$. 
		As a result, 
		\[
		1 = \sum_{j\in N}v(A_j) > \alpha + \alpha + (n-2)\cdot\frac{\alpha}{2} = \frac{n+2}{2} \alpha > 1,
		\]
		which is also a contradiction. 
		
		For the other direction,
		the disutility function for this subcase also
		contains one object with disutility $\alpha$ and $n$ objects with disutility $\frac{1-\alpha}{n}$ (see Table \ref{tab:k=0-D}). 
		Since $\frac{2}{n+2} < \alpha \le 1$, it follows that $2\cdot \frac{1-\alpha}{n} < \alpha \le 1$. 
		Clearly, the MinMaxShare of this disutility function is $\alpha = \Delta_n^{\oplus}(\alpha)$. Up to here, the proof regarding the case of $k=0$ is completed.
		
		\medskip
		
		Next, we consider the general case of $ k\ge 1$ excluding $n=2$ and $k=1$.

		\medskip
		
		\noindent{\bf Case 2: $n\ge 3$ and $k\ge 1$ or $n\ge 2$ and $k\ge 2$}
		
		\medskip
		
		For this case, we again start with the subcases when $\alpha \in D(n,k)$. 
		Recall that when $\alpha \in D(n,k)$, $\alpha \in (\frac{1}{(k+1)n+1}, \frac{k+2}{n(k+1)^2+k+2}]$ and $v(A_1) > \Delta_n^{\oplus}(\alpha)=\frac{k+2}{k+1} \cdot \frac{1-\alpha}{n}$.
		
		\medskip
		
		{\bf Subcase 2.1: $\alpha \in D(n,k)$ and $E(\alpha) \cap A_j = \emptyset$ for any $j \in N \setminus \{1\}$}

		\medskip

		
		In this subcase, all the objects with disutility $\alpha$ are in $A_1$,
		and thus $v(e) < \alpha$ for any $e \in A_j$ and $j \in N \setminus \{1\}$. 
		Due to the normalization, there exists an agent $j_0$ who receives disutility at most $\frac{1-v(A_1)}{n-1}$,
		which gives the following lower bound of the difference between the disutilities that agents $1$ and $j_0$ receive
		\[
		v(A_1) - v(A_{j_0}) \ge \frac{n}{n-1}v(A_1) - \frac{1}{n-1} > \frac{1-(k+2)\alpha}{(n-1)(k+1)}. 
		\] 
		It can be shown that the rightmost-hand side of the above inequality is no less than $\frac{\alpha}{2}$, which is equivalent to $\alpha \le \frac{2}{(k+1)n+k+3}$. 
		Since $\alpha \le \frac{k+2}{n(k+1)^2+k+2}$, it suffices to show $\frac{2}{(k+1)n+k+3} \ge \frac{k+2}{n(k+1)^2+k+2}$,
		which holds since
		\[
		\frac{2}{(k+1)n+k+3} - \frac{k+2}{n(k+1)^2+k+2} = \frac{(k+1)(nk-k-2)}{((k+1)n+k+3)(n(k+1)^2+k+2)} \ge 0,
		\]
		where the last inequality is because $n\ge 3$ and $k\ge 1$, or $n\ge 2$ and $k\ge 2$. 
		
		Therefore, $v(A_1) - v(A_{j_0}) > \frac{\alpha}{2}$. 
		Let $e^*$ be an object in $A_1$ with disutility $\alpha$. 
		Since $v(A_1) > (k+1)\alpha > \alpha$, for any $S \subseteq A_{j_0}$ with disutility smaller than $\alpha$, 
		Claim \ref{clm:diff} actually gives a tighter bound of its disutility, i.e., $v(S) \le v(e^*) - (v(A_1) - v(A_{j_0})) < \frac{\alpha}{2}$. 
		Thus, $v(e) < \frac{\alpha}{2}$ for any $e \in A_{j_0}$.
		Besides, according to Claim \ref{clm:v_Aj}, $v(A_{j_0}) \ge v(e^*) = \alpha$. 
		These two facts together imply that there exists $S' \subseteq A_{j_0}$ such that $v(S') \in [\frac{\alpha}{2}, \alpha)$, which is a contradiction to Claim \ref{clm:diff}.
		
		\medskip
		
		
		{\bf Subcase 2.2: $\alpha \in D(n, k)$ and $E(\alpha) \cap A_{j'} \neq \emptyset$ for some $j' \in N \setminus \{1\}$}
		
		\medskip
		
		In this subcase, some objects with disutility $\alpha$ are in $A_{j'}$. 
		Before diving into the proof for this subcase, we present the following claim, 
		which shows the existence of a subset of $A_1$ whose disutility is within a specific range. 
		
		\begin{claim} \label{claim:subset}
			There exists a subset $S \subseteq A_1$ such that $\frac{k}{k+2}v(A_1) \le v(S) < v(A_1) - \alpha$. 
		\end{claim}
		\begin{proof}[Proof of Claim \ref{claim:subset}]
			When $k = 1$, if there exists $e \in A_1$ such that $v(e) \ge \frac{1}{3}v(A_1)$, recall that $v(A_1) > (k+1)\alpha = 2\alpha$, Claim \ref{claim:subset} holds since $v(e) \le \alpha < v(A_1) - \alpha$. 
			If $v(e) < \frac{1}{3}v(A_1)$ for any $e \in A_1$, denote by $(A_1^1, A_1^2)$ one 2-partition of $A_1$ that minimises the disutility difference between the two bundles among all the 2-partitions. 
			Without loss of generality, we assume $v(A_1^1) \le v(A_1^2)$, then $v(A_1^1) \le \frac{1}{2}v(A_1) < v(A_1) - \alpha$. 
			Besides, $v(A_1^1) \ge \frac{1}{3}v(A_1)$ holds. 
			Otherwise, $v(A_1^2) - v(A_1^1) = v(A_1) - 2v(A_1^1) > \frac{1}{3}v(A_1)$, implying that moving an object from $A_1^2$ to $A_1^1$ returns another 2-partition of $A_1$ that has a smaller disutility difference, which contradicts the definition of $(A_1^1, A_1^2)$. 
			
			When $k \ge 2$, we first show that $v(e) > \frac{1}{k+2}\alpha$ for any $e \in A_1$. 
			If not, $v(A_1) > v(A_1 \setminus \{e\}) \ge v(A_1) - \frac{1}{k+2}\alpha$. 
			Then Claim \ref{clm:v_Aj} gives $v(A_j) \ge v(A_1) - \frac{1}{k+2}\alpha$ for any $j \in N \setminus \{1\}$. 
			Summing up these lower bounds gives the following inequality
			\[
			1 = \sum_{j \in N} v(A_j) \ge v(A_1) + (n-1)v(A_1) - \frac{n-1}{k+2}\alpha > \frac{k+2}{k+1} - \frac{(k+2)^2 + (k+1)(n-1)}{(k+1)(k+2)}\alpha.
			\]
			It can be shown that the rightmost-hand side is at least 1, which constitutes a contradiction. 
			This is equivalent to show that $\alpha \le \frac{k+2}{(k+2)^2 + (k+1)(n-1)}$. 
			Since $\alpha \le \frac{k+2}{n(k+1)^2+k+2}$, it suffices to show that $\frac{k+2}{(k+1)^2 + (k+1)(n-1)} \ge \frac{k+2}{n(k+1)^2+k+2}$,
			which holds since 
			\[
			n(k+1)^2+k+2 - ((k+2)^2 + (k+1)(n-1)) = (k+1)(nk - k -1) \ge 0,
			\]
			where the last inequality is because $n \ge 2$ and $k \ge 1$. 
			
			We then let $S^* = \arg \min_{S \subseteq A_1, v(S) > \alpha}v(S)$ which is guaranteed to exist since $v(A_1) > (k+1)\alpha > \alpha$, and show by contradiction that $v(S^*) \le \frac{2}{k+2}v(A_1)$. 
			This gives $\frac{k}{k+2}v(A_1) \le v(A_1 \setminus S^*) < v(A_1) - \alpha$. 
			We assume for the sake of contradiction that $v(S^*) > \frac{2}{k+2}v(A_1)$. 
			Then the definition of $S^*$ gives the following lower bound of $v(e)$ for any $e \in S^*$
			\[
			v(e) > v(S^*) - \alpha > \frac{2}{k+2}v(A_1) - \alpha > \frac{k}{k+2}\alpha \ge \frac{1}{2}\alpha,
			\]
			where the second last inequality is because $v(A_1) > (k+1)\alpha$ and the last inequality is because $k \ge 2$. 
			This lower bound implies that $S^*$ contains exactly 2 objects. 
			Otherwise (i.e., $|S^*| \ge 3$), for any subset $S' \subseteq S^*$ that contains exactly 2 objects, $\alpha < v(S') < v(S^*)$ holds, which contradicts the definition of $S^*$. 
			
			Therefore, we can denote $S^* = \{e^l, e^s\}$ and assume without loss of generality that $v(e^l) \ge v(e^s)$. 
			Accordingly, $v(e^l) \ge \frac{1}{2}v(S^*) > \frac{1}{k+2}v(A_1) > \frac{k+1}{k+2}\alpha$. 
			Recall that $v(e) > \frac{1}{k+2}\alpha$ holds for any $e \in A_1$.
			These two facts together imply that the total disutility of $e^l$ and any other object in $A_1$ is larger than $\alpha$.
			From the definition of $S^*$, we know that $e^s \in \arg \min_{e \in A_1}v(e)$,
			which gives $v(e) \ge v(e^s) > \frac{k}{k+2}\alpha$ for any $e \in A_1$. 
			Let $S'$ be the subset of $A_1$ that contains the two objects with the smallest disutilities, 
			the following inequality leads to a contradiction to the definition of $S^*$
			\[
			\alpha \le \frac{2k}{k+2}\alpha < v(S') \le \frac{2}{k+2}v(A_1) < v(S^*), 
			\]
			where the first inequality is because $k \ge 2$ and the second last inequality is because $|A_1| \ge k+2$. 
		\end{proof}
		
		We are now ready to reveal the contradiction in the subcase. 
		Denote by $e^*$ one object in $A_{j'}$ that has disutility $\alpha$ and by $S$ a subset of $A_1$ that satisfies Claim \ref{claim:subset},
		Claim \ref{clm:v_Aj} gives $v(A_{j'} \setminus \{e^*\}) \ge v(S) \ge \frac{k}{k+2}v(A_1)$; that is, $v(A_{j'}) \ge \frac{k}{k+2}v(A_1) + \alpha$. 
		For any $j \in N \setminus \{1, j'\}$, recall that $|A_1| \ge k+2$ which implies that there exists $S' \subseteq A_1$ such that $v(A_1) > v(S') \ge \frac{k+1}{k+2}v(A_1)$,
		Claim \ref{clm:v_Aj} gives $v(A_j) \ge v(S') \ge \frac{k+1}{k+2}v(A_1)$. 
		Summing up these lower bounds leads to the following contradiction
		\begin{align*}
			1 = \sum_{j \in N} v(A_j) &\ge v(A_1) + \frac{k}{k+2}v(A_1) + \alpha + (n-2)\cdot \frac{k+1}{k+2}v(A_1) \\
			&= \frac{n(k+1)}{k+2}v(A_1) + \alpha > 1 - \alpha + \alpha = 1.
		\end{align*}
		
		For the other direction, 
		the disutility function for the subcases when $\alpha \in D(n, k)$ 
		(see Table \ref{tab:general-D}) containing one object with disutility $\alpha$ and $n(k+1)$ objects with disutility $\frac{1-\alpha}{n(k+1)}$. 
		Since $\alpha > \frac{1}{kn+n+1}$, it follows that $\alpha > \frac{1-\alpha}{n(k+1)}$. 
		Besides, it can be verified that $\alpha < \frac{2-2\alpha}{n(k+1)}$, which is equivalent to $\alpha < \frac{2}{nk+n+2}$.
		Since $\alpha \le \frac{k+2}{n(k+1)^2+k+2}$, it suffices to show $\frac{k+2}{n(k+1)^2+k+2} < \frac{2}{n(k+1)+2}$, which holds since
		\[
		\frac{2}{n(k+1)+2} - \frac{k+2}{n(k+1)^2+k+2} = \frac{nk(k+1)}{(n(k+1)+2)(n(k+1)^2+k+2)} > 0. 
		\]
		By the pigeonhole principle, there exists a bundle that contains at least $k+2$ objects in any allocation. 
		This implies that the MinMaxShare of this disutility function is $(k+2)\cdot \frac{1-\alpha}{n(k+1)}$, 
		which happens in the allocation where one bundle contains $k+2$ objects with disutility $\frac{1-\alpha}{n(k+1)}$, 
		one bundle contains $k$ objects with disutility $\frac{1-\alpha}{n(k+1)}$ and one object with disutility $\alpha$,
		and each of the other bundles contains $k+1$ objects with disutility $\frac{1-\alpha}{n(k+1)}$. 
		
		\begin{table}[H]
			\centering
			\begin{tabular}{c|c}
				\hline
				Object Disutility & Quantity \\
				\hline
				$\alpha$ & 1 \\
				$\frac{1-\alpha}{n(k+1)}$ & $n(k+1)$ \\
				\hline
			\end{tabular}
			\caption{Disutility function for subcases $\alpha \in D(n, k)$ with $n\ge 3$ and $k\ge 1$, or $n\ge 2$ and $k\ge 2$.}
			\label{tab:general-D}
		\end{table}
		
		Next we consider the subcases when $\alpha \in I(n,k)$. 
		Recall that when $\alpha \in I(n, k)$, $\alpha \in (\frac{k+2}{n(k+1)^2+k+2}, \frac{1}{kn+1}]$ and $v(A_1) > \Delta_n^{\oplus}(\alpha)=(k+1)\alpha$.  
		
		\medskip
		
		{\bf Subcase 2.3:  $\alpha \in I(n,k)$ and $E(\alpha) \cap A_j = \emptyset$ for any $j \in N \setminus \{1\}$}
		
		\medskip
		
		For this subcase, we first derive a lower bound of $v(A_j)$ for any $j \in N \setminus \{1\}$, i.e., $v(A_j) \ge (\frac{(k+1)^2}{k+2} + \frac{1}{(n-1)(k+2)})\alpha$. 
		Letting $D = \frac{(k+1)^2}{k+2} + \frac{1}{(n-1)(k+2)}$,
		we assume for the sake of contradiction that $v(A_{j'}) < D\alpha$ for some $j' \in N \setminus \{1\}$. 
		It can be verified that $k < D < k+1$, where the first inequality is equivalent to $n > 0$, and the second inequality is equivalent to $(n-1)(k+1) > 1$.
		Denote by $e^*$ one object in $A_1$ with disutility $\alpha$ and by $e'$ any object in $A_j$, we have
		\[
		v(A_{j'} \setminus (A_{j'} \setminus \{e'\}) \cup (A_1 \setminus \{e^*\})) = v(A_1 \setminus \{e^*\} \cup \{e'\}) < v(A_1). 
		\]
		Then from Claim \ref{clm:v_Aj}, $v(A_{j'} \setminus \{e'\}) \ge v(A_1 \setminus \{e^*\})$, which gives
		\[
		v(e') \le v(A_{j'}) - v(A_1) + v(e^*) < D\alpha - (k+1)\alpha + \alpha = (D-k)\alpha. 
		\]
		However, we next show that the disutility of some object in $A_{j'}$ must be larger than $(D-k)\alpha$, which leads to a contradiction. 
		To achieve this, we denote $S^* \in \arg \min_{S\subseteq A_{j'}, v(S) > (D-1)\alpha}v(S)$,
		whose existence is guaranteed since Claim \ref{clm:v_Aj} gives $v(A_{j'}) \ge v(A_1 \setminus \{e^*\}) > k\alpha > (D-1)\alpha$. 
		Notice that
		\[
		v(A_{j'} \setminus S^* \cup (A_1 \setminus \{e^*\})) < D\alpha - (D-1)\alpha + v(A_1) - \alpha = v(A_1), 
		\]
		from Claim \ref{clm:v_Aj}, $v(S^*) \ge v(A_1 \setminus \{e^*\}) > k\alpha$. 
		Then the definition of $S^*$ implies that the disutility of any object in $S^*$ is at least 
		\[
		v(S^*) - (D-1)\alpha > (k-D+1)\alpha \ge (D-k)\alpha,
		\]
		where the last inequality is equivalent to $D-k-\frac{1}{2} = \frac{k+2-kn}{2(n-1)(k+2)} \le 0$, which holds when $n\ge 3$ and $k\ge 1$, or $n\ge 2$ and $k\ge 2$. 
		
		Therefore, $v(A_j) \ge (\frac{(k+1)^2}{k+2} + \frac{1}{(n-1)(k+2)})\alpha$ holds for any $j \in N \setminus \{1\}$.
		Summing up these lower bounds leads to the following contradiction
		\[
		1 = \sum_{j \in N}v(A_j) > (k+1)\alpha + (n-1)\cdot (\frac{(k+1)^2}{k+2} + \frac{1}{(n-1)(k+2)})\alpha = \frac{n(k+1)^2+k+2}{k+2}\alpha > 1. 
		\]
		
		\medskip
		
		\textbf{Subcase 2.4: $\alpha \in I(n,k)$ and $E(\alpha) \cap A_{j'} \neq \emptyset$ for some $j' \in N \setminus \{1\}$}
		
		\medskip
		
		The proof is similar to that of Subcase 2.2. 
		First, it can be verified that Claim \ref{claim:subset} still holds. 
		
		
		\noindent\begin{proof}[Proof of Claim \ref{claim:subset} for $\alpha \in I(n,k)$]
			Notice that Claim \ref{claim:subset} holds as long as $k=1$, thus, we can focus on $k\ge 2$. 
			We first show that $v(e) > \frac{1}{k+2}\alpha$ for any $e \in A_1$. 
			If not, $v(A_1 \setminus \{e\}) \ge v(A_1) - \frac{1}{k+2}\alpha$. 
			Then Claim \ref{clm:v_Aj} gives $v(A_j) \ge v(A_1) - \frac{1}{k+2}\alpha$ for any $j \in N \setminus \{1\}$. 
			Summing up these lower bounds gives the following formula
			\[
			1 = \sum_{j \in N} v(A_j) \ge v(A_1) + (n-1)v(A_1) - \frac{n-1}{k+2}\alpha > \frac{n(k+1)(k+2)-n+1}{k+2}\alpha.
			\]
			It can be shown that the rightmost-hand side of the above inequality is at least 1, which is a contradiction. 
			This is equivalent to show that $\alpha \ge \frac{k+2}{n(k+1)(k+2)-n+1}$. 
			Since $\alpha \ge \frac{k+2}{n(k+1)^2+k+2}$, it suffices to show that $\frac{k+2}{n(k+1)(k+2)-n+1} \le \frac{k+2}{n(k+1)^2+k+2}$,
			which holds since 
			\[
			n(k+1)(k+2)-n+1 - (n(k+1)^2+k+2) = nk - k - 1 \ge 0
			\]
			where the last inequality is because $n \ge 2$ and $k \ge 1$. 
			
			We then let $S^* = \arg \min_{S \subseteq A_1, v(S) > \alpha}v(S)$, which is guaranteed to exist since $v(A_1) > (k+1)\alpha > \alpha$.
			By the same proof as the counterpart in the proof of Claim \ref{claim:subset} for $\alpha \in D(n,k)$, we can show that $v(S^*) \le \frac{2}{k+2}v(A_1)$, which gives $\frac{k}{k+2}v(A_1) \le v(A_1 \setminus S^*) < v(A_1) - \alpha$. 
		\end{proof}
		
		\medskip

		We are now ready to reveal the contradiction in this subcase. 
		Denote by $e^*$ one object in $A_{j'}$ that has disutility $\alpha$ and by $S$ a subset of $A_1$ that satisfies Claim \ref{claim:subset},
		Claim \ref{clm:v_Aj} gives $v(A_{j'} \setminus \{e^*\}) \ge v(S) \ge \frac{k}{k+2}v(A_1)$; that is, $v(A_{j'}) \ge \frac{k}{k+2}v(A_1) + \alpha$. 
		For any $j \in N \setminus \{1, j'\}$, recall that $|A_1| \ge k+2$ which implies that there exists $S' \subseteq A_1$ such that $v(A_1) > v(S') \ge \frac{k+1}{k+2}v(A_1)$,
		Claim \ref{clm:v_Aj} gives $v(A_j) \ge v(S') \ge \frac{k+1}{k+2}v(A_1)$. 
		Summing up these lower bounds leads to the following contradiction
		\begin{align*}
			1 = \sum_{j \in N} v(A_j) &\ge v(A_1) + \frac{k}{k+2}v(A_1) + \alpha + (n-2)\cdot \frac{k+1}{k+2}v(A_1) \\
			&= \frac{n(k+1)}{k+2}v(A_1) + \alpha > \frac{n(k+1)^2+k+2}{k+2}\alpha > 1.
		\end{align*}
		
		For the other direction, 
		the disutility function for the subcases when $\alpha \in I(n, k)$ 
		(See Table \ref{tab:general-I}) containing $kn+1$ objects with disutility $\alpha$ and $n-1$ objects with disutility $\frac{1-(nk+1)\alpha}{n-1}$. 
		It can be verified that $\alpha > \frac{1-(kn+1)\alpha}{n-1}$, which is equivalent to $\alpha > \frac{1}{(k+1)n}$.
		Since $\alpha > \frac{k+2}{n(k+1)^2+k+2}$, it suffices to show $\frac{k+2}{n(k+1)^2+k+2} \ge \frac{1}{(k+1)n}$, which holds since
		\[
		\frac{k+2}{n(k+1)^2+k+2} - \frac{1}{(k+1)n} = \frac{(k+1)n - k - 2}{(n(k+1)^2+k+2)(k+1)n} \ge 0,
		\]
		where the inequality is because $n\ge 3$ and $k\ge 1$, or $n\ge 2$ and $k\ge 2$. 
		By the pigeonhole principle, there exists a bundle that contains at least $k+1$ objects with disutility $\alpha$. 
		This implies that the MinMaxShare of this disutility function is $(k+1)\alpha$,
		which happens in the allocation where one bundle contains $k+1$ objects with disutility $\alpha$, and each of the other bundles contains $k$ objects with disutility $\alpha$ and one object with disutility $\frac{1-(nk+1)\alpha}{n-1}$. 
		
		\begin{table}[H]
			\centering
			\begin{tabular}{c|c}
				\hline
				Object Disutility & Quantity \\
				\hline
				$\alpha$ & $kn+1$ \\
				$\frac{1-(nk+1)\alpha}{n-1}$ & $n-1$\\
				\hline
			\end{tabular}
			\caption{Disutility function for subcases $\alpha \in I(n, k)$ with $n\ge 3$ and $k\ge 1$, or $n\ge 2$ and $k\ge 2$.}
			\label{tab:general-I}
		\end{table}
		
		Up to here, we computed Hill's share
		for unrestricted $m$, except for $n=2$ and $k=1$.
		Moving to the setting of restricted $m$,
		Hill's share can be computed by
		similar approaches with a more involved discussion.
		The remaining proof of Theorem~\ref{thm:homogeneous} is regulated to the appendix.

		\section{Hill's Guarantee for Indivisible Bads}
		
		\subsection{Main Result}
		
		In this section, we prove the counterpart result of Hill's guarantee for indivisible bads.
		Consider the general case, where each one of the $n$ agents now has an arbitrary disutility $v_{i}$ in $%
		\mathcal{A}dd(M)$ (by convention $m=|M|$). 
		Given $m$ and $n$%
		, a \textit{guarantee} specifies an upper bound $\Gamma _{n}(v_{i};m)$ on
		agent $i$'s disutility when she shares the $m$ bads with $n-1$ other agents
		of unknown disutilities in $\mathcal{A}dd(M)$. By construction the mapping $%
		\Gamma _{n}$ is the same for every agent $i$. As part of its definition, a
		guarantee must be feasible: for any profile $(v_{i})_{i=1}^{n}\in \lbrack 
		\mathcal{A}dd(M)]^{n}$ there exists an allocation $(A_{1},\ldots ,A_{n})$
		of $M$ such that%
		\begin{equation}
			v_{i}(A_{i})\leq \Gamma _{n}(v_{i};m)\mbox{ for all } 1\leq i\leq n.
			\label{2}
		\end{equation}
		
		We know from \cite{DBLP:conf/aaai/AzizRSW17} and \cite{DBLP:conf/wine/FeigeST21} that the MinMaxShare value $%
		\MMS_{n}(v_{i})$ is not a guarantee because at some (rare!) profiles no
		allocation meets all inequalities in (\ref{2}). By applying Inequalities~(\ref{2})
		to an arbitrary guarantee $\Gamma _{n}$ at the unanimous profile $v_{i}=v$
		for all $i$, we see it is lower bounded by the MinMaxShare:%
		\[
		\Gamma _{n}(v;m)\geq \MMS_{n}(v)\mbox{ for all }v\in \mathcal{A}dd(M).
		\]
		
		In this section, we show that the monotone hull of $\Delta_n^{\oplus}$ serves as the best guarantee in Hill's model.
		Recall that $\cU(\alpha;m)$ contains all the disutility functions $v(\cdot)$ on objects $M$ such that $\max_{e\in M}v(e) \le \alpha$, 
		and $\cU(\alpha) = \bigcup_m \cU(\alpha;m)$.
		For simplicity in the presentation and analysis, we ignore the restriction of the number of objects $m$ in this section, 
		and the result can be extended to the setting with parameter $m$ using the same approach in Section \ref{sec::upper-bound-minmaxshare}.
		The definition of $\cU(\alpha)$ is the same as in
		\citep{hill1987partitioning,DBLP:conf/wine/MarkakisP11,DBLP:journals/tcs/GourvesMT15}.
		Note that $\cU(\alpha') \subseteq \cU(\alpha)$ if $\alpha'\le \alpha$, and the difference between $\cV(\alpha)$ and $\cU(\alpha)$ is that the disutilities in $\cU(\alpha)$ do not require that there must be one object with disutility $\alpha$. 
		It is straightforward that the tight guarantee regarding $\cU(\cdot)$ must be monotone non-decreasing since any worst-case disutility in $\cU(\beta)$ is also a disutility in $\cU(\alpha)$ for $\beta \le \alpha$.
		We write $V_{n}$ the monotone hull of $\Delta_n^{\oplus}$
		\[
		V_{n}(\alpha )=\max_{0\leq \beta \leq \alpha }\Delta_n^{\oplus}(\beta ),
		\]%
		as illustrated in Fig. \ref{fig:result1:hete} when $n=2,3$. In more detail, we have the following formula of $V_{n}$:
		
			$$
			V _ n ( \alpha) = 
			\left\{\begin{array}{ll} 
				\frac{k+2}{(k+1)n + 1}, & \text { if } 
				\alpha \in NI(n,k) \\
				(k+1)\alpha, & \text { if } \alpha \in I(n,k)
			\end{array}\right.
			$$
			where for any integer $k \geq 0 $,
			$$
			NI(n,k) = \left( \frac{1}{(k+1)n + 1}, \frac{k+2}{(k+1)((k+1)n + 1)} \right)
			$$
			and
			$$
			I(n,k) = \left[ \frac{k+2}{(k+1)((k+1) n + 1)}, \frac{1}{kn+1} \right].
			$$
			
			
			\begin{figure}
				\centering
				\includegraphics[width=0.55\textwidth]{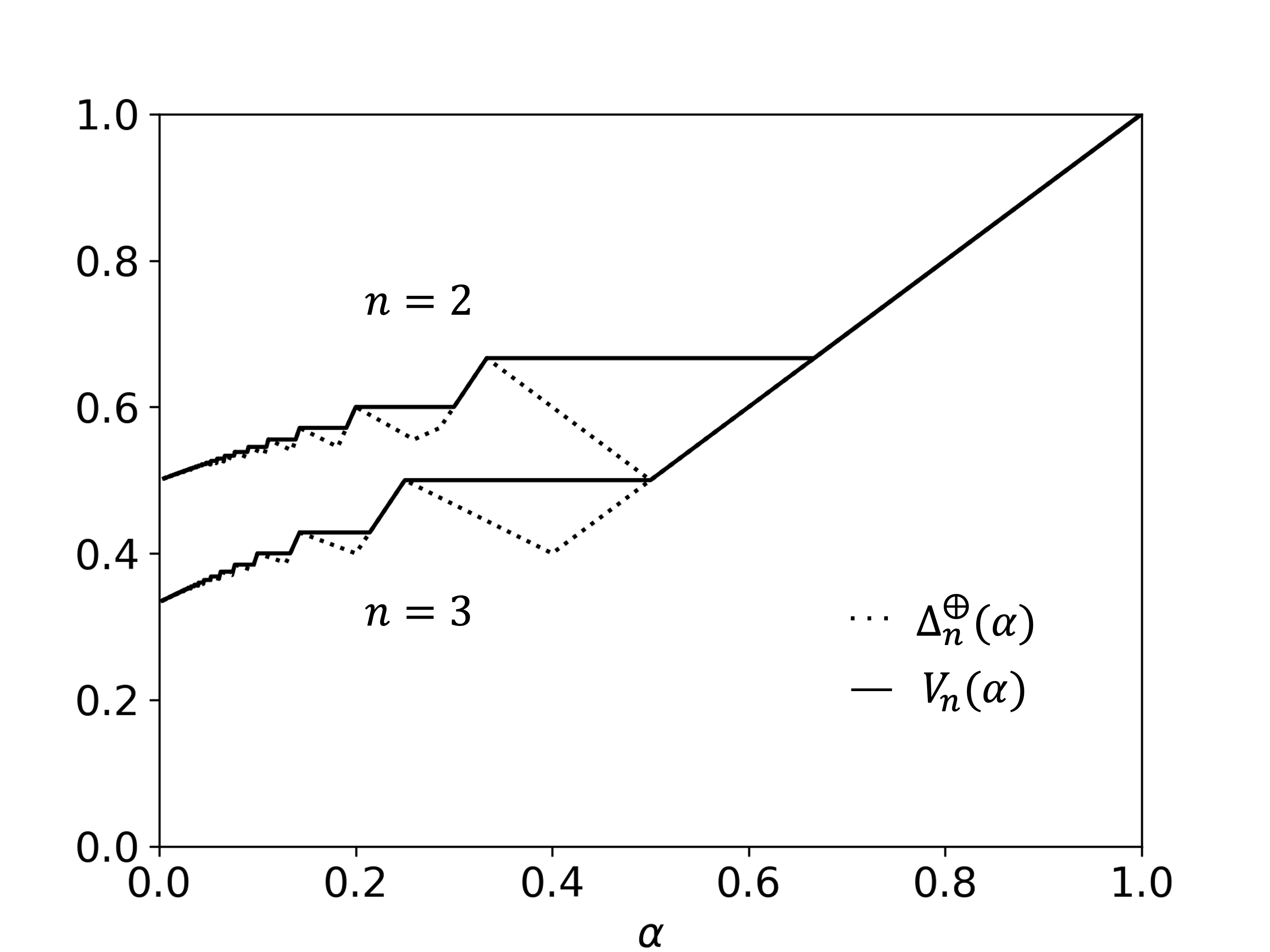}
				\caption{The characterisation for Heterogeneous Agents 
				}
				\label{fig:result1:hete}
			\end{figure}
			
			By Theorem \ref{thm:homogeneous} and the construction of $V_n(\cdot)$, $V_n(\cdot)$ provides the tight bound
			of the worst-case MinMaxShare regarding $\cU(\cdot)$. 
			We further prove that {$V_n(\cdot)$ is a guarantee and moreover}
			an allocation satisfying
			$V_n(\cdot)$ can be found in polynomial time.

			\begin{theorem} \label{thm:heterogeneous}
				$\Gamma _{n}(v)=V_{n}(\max_{e\in M}v(e))$ defines a canonical guarantee.
				That is, given any $0\le \alpha_i \le 1$ and $v_i\in \cU(\alpha_i)$ for $i = 1,\ldots, n$, there exists an allocation $(A_1,\ldots,A_n)$ with
				\[
				v_i(A_i) \le V_n(\alpha_i) \text{ for all $i = 1,\ldots, n$}
				\]
				and such an allocation
				can be computed in polynomial time. 
				\textcolor{black}{Moreover, for any $0\leq \alpha \leq 1$,
					there exists $\{v'_i\}_{i=1}^n$ with $v ^ {\prime} _ i \in \cU(\alpha)$ for any $i \in [n]$ such that $V_n(\alpha)$ is the best possible guarantee}, i.e., for any allocation $(B_1,\ldots,B_n)$, 
				\[
				\text{there exists } i\in N \text{ such that } v'_i(B_i) \ge V_n(\alpha).
				\]
			\end{theorem}
			
			As for $\Delta_n^{\oplus}(\cdot)$ in Theorem 1 the two key features of this guarantee are: its
			computation is elementary and it does not depend on
			the number of bads to allocate.
			As far as we know, no other similarly simple guarantee for the allocation of bads has been
			identified.
			

			\begin{remark}\label{remark:heter}
				By Theorem \ref{thm:heterogeneous}, $V_{n}(\alpha)$ is the best guarantee for disutilities in $\cU(\alpha)$,  and thus we get the tight counterpart result of \cite{hill1987partitioning} for bads. 
				However, it may not be the best in the model of \citet{DBLP:journals/tcs/GourvesMT15}, i.e., for disutilities in $\cV(\alpha)$.
				For example, when $n=2$, we can show that $\Delta_2^{\oplus}(\max_{e\in M}v_i(e))$ is a tighter guarantee in the later model.
				Given two disutility functions $v_1$ and $v_2$, without loss of generality, suppose $\Delta_2^{\oplus}(\max_{e} v_1(e)) \le \Delta_2^{\oplus}(\max_{e} v_2(e))$.
				Then we find the MinMax partition of $v_1$ 
				so that the disutilities of both bundles are no greater than $\Delta_2^{\oplus}(\max_{e} v_1(e))$ to agent 1. 
				We ask agent 2 to choose a better bundle whose disutility must be no greater than $\frac{1}{2}$ and thus no greater than $\Delta_2^{\oplus}(\max_{e} v_2(e))$ to agent 2.
				It is still open whether $\Delta_n^{\oplus}(\max_{e\in M}v_i(e))$ is a guarantee or not when $n\ge 3$ in \citet{DBLP:journals/tcs/GourvesMT15}'s model, which is an interesting future research direction. 
			\end{remark}

			\subsection{Proof of Theorem \ref{thm:heterogeneous}}
			
			To show that one can compute an allocation satisfying the required bound in Theorem \ref{thm:heterogeneous}, we derive a variation of the moving-knife algorithm. 
			When the objects are goods and divisible, \citet{dubins1961cut} proved that such an algorithm (also known as Dubins-Spanier moving knife algorithm) gives the optimal worst-case bound, i.e., every agent gets value for at least $\frac{1}{n}$.
			\citet{DBLP:conf/wine/MarkakisP11} further proved that a variation of this algorithm also guarantees the optimal worst-case bound for indivisible goods.
			In a nutshell, towards proving Theorem \ref{thm:heterogeneous}, we first use the reduction proved in \citep{DBLP:journals/aamas/BouveretL16,DBLP:conf/sigecom/HuangL21} to restrict our attention to the ordered instances when agents have the same ranking over all objects,
			which significantly simplifies our analysis.
			Then we show that using $V_n(\cdot)$ to set the parameters in the moving-knife algorithm
			always returns an allocation ensuring the bound in Theorem \ref{thm:heterogeneous}.
			
			The following lemma says that it suffices to only focus on the ordered instances. 
			
			\begin{lemma}[\citep{DBLP:journals/aamas/BouveretL16,DBLP:conf/sigecom/HuangL21}]
				Suppose there is an algorithm that takes any ordered instance as input, runs in $T(n, m)$ time and returns an allocation where each agent $ i $'s disutility is at most $ V_n(\alpha_i)$.  
				Then, we have an algorithm that takes any instance as input, runs in $T(n, m) + O(nm \log m)$ time and returns an allocation with the same disutility guarantee. 
			\end{lemma}
			
			Our approach is similar to that in \cite{DBLP:conf/wine/MarkakisP11}, but the detailed
			proof differs.
			Our algorithm runs in recursions. In each recursion, the algorithm allocates a bundle of objects to one agent in a moving-knife fashion. 
			Each time, each of the remaining agents moves her “knife” one object towards the objects with smaller disutilities, 
			until for every agent $i$ the total disutility of the objects before her "knife" is larger than $V_n(\alpha_i)$. 
			After that, one of the last agents (denoted by agent $k$) for whom the total utility of the objects before her knife is larger than $V_n(\alpha_k)$ receives the objects except the one right before her knife. 
			If there remains only one agent who has not received a bundle, she will get all the remaining objects. 
			Otherwise, all the remaining agents enter the next recursion with their disutility functions being normalised such that for each of them the total disutility of the remaining objects is 1. 
			The formal description of our algorithm is presented in Algorithm \ref{alg:hetero}. 
			
			\begin{algorithm}[htbp]
				\caption{Algorithm for heterogeneous disutilities}
				\label{alg:hetero}
				\begin{algorithmic}[1]
					\REQUIRE An ordered instance with agents $N$, objects $M$ and disutility functions $\{v_i\}_{i \in N}$. 
					\ENSURE An allocation $\mathbf{A} = \{A_1, \ldots, A_n\}$ with $v_i(A_i) \le V_n(\alpha _ i )$ for every $i \in N$.
					\STATE Initialize $ S _ i  = \emptyset$ for every $ i \in N$.
					\WHILE{there exists an agent $j$ with $v _ j  ( S _ j ) \leq V _ n (\alpha _ j )$}
					\FOR {every $i \in N$}
					\STATE $ S _ i  \leftarrow S _ i  \cup\{\text{the object  in $M \setminus S_i$ with the largest disutility for agent $i$ (tie breaks arbitrarily)}\}$. 
					\ENDFOR
					\ENDWHILE
					\STATE Pick the agent $k \in N$ with $ v_ k ( S _ k \setminus \{ \widetilde{e} \}) \leq V _ n (\alpha _ k )$ where $\widetilde{e} $ is the last object that $k$ added into $S_k$ (tie breaks arbitrarily). 
					\STATE $A_k \leftarrow S_k \setminus \{ \widetilde{e} \}$. 
					\IF{ $ | {N} | = 2 $}
					\STATE Allocate $M \setminus A_k$ to the remaining agent.
					\ELSE
					\STATE Construct a new disutility function $v'_i$ for every $i \in N \setminus \{k\}$ by setting $v'_i(e) = \frac{v_i(e)}{1-v_i(A_k)}$ for every $e \in M \setminus A_k$. 
					\STATE Run Algorithm \ref{alg:hetero}($N \setminus \{k\}$, $M \setminus A_k$, $\{v'_i\}_{i \in N \setminus \{k\}}$). 
					\ENDIF
				\end{algorithmic}
			\end{algorithm}
			
			Then we are going to prove Theorem \ref{thm:heterogeneous}.
			Without loss of generality, let $1, \ldots, n$ be the order in which agents receive bundles in Algorithm \ref{alg:hetero}. 
			Denote $C_i = v_i(A_1)$ for every $N \setminus \{1\}$, the following lemma gives a lower bound of $C_i$. 
			
			\begin{lemma} \label{lem:Ci}
				For any agent $i \in N \setminus \{1\}$ with $\alpha_i \in NI(n,k) \cup I(n,k)$ for some $k \ge 0$,
				we have 
				\[
				C _ i  \ge  \frac{1 - V _ n (\alpha _ i )}{ n- 1 }.
				\]
			\end{lemma}
			\begin{proof}
				Denote by $q$ the index such that $\sum_{e=1}^q v_i(e) \le V_n(\alpha_i)$ and $\sum_{e=1}^{q+1} v_i(e) > V_n(\alpha_i)$, whose existence is guaranteed since $v_i(M) > V_n(\alpha_i)$.
				Since $V_n(\alpha_i) \ge (k+1)\alpha_i$ (this can be easily verified from the definition of $V_n(\alpha)$ and can also be seen from Fig. \ref{fig:result1:hete}) and $v_i(e) \le \alpha_i$ for every $e \in M$, $q \ge k+1$. 
				Otherwise, $\sum_{e=1}^{q+1} v_i(e) \le (k+1)\alpha_i \le V_n(\alpha_i)$, which contradicts the definition of $q$. 
				According to Algorithm \ref{alg:hetero}, $C_i \ge \sum_{e=1}^q v_i(e)$. 
				Since only ordered instances are considered, $v_i(\{q+1\}) \le v_i(\{q\}) \le \frac{C_i}{k+1}$, which gives
				\[
				C_i + \frac{C_i}{k+1} \ge \sum_{e=1}^{q+1} v_i(e) > V_n(\alpha_i).
				\]
				Therefore, $C_i > \frac{k+1}{k+2} \cdot V_n(\alpha_i)$. 
				We consider the following two cases regarding the ranges of $\alpha_i$. 
				
				\medskip
				
				\textbf{Case 1}: $\alpha_i \in I(n,k)$. In this case, $\frac{k+2}{(k+1)((k+1) n + 1)} \le \alpha_i \le \frac{1}{kn+1}$ and $V_n(\alpha_i) = (k+1)\alpha_i$.
				Then, 
				\[
				C_i > \frac{k+1}{k+2}\cdot V_n(\alpha) \ge \frac{1 - V_n(\alpha)}{n-1},
				\]
				where the last inequality holds since $\alpha_i \ge \frac{k+2}{(k+1)((k+1)n + 1)}$. 
				
				\medskip
				
				\textbf{Case 2}: $\alpha_i \in NI(n,k)$. In this case, $V_n(\alpha_i) = \frac{k+2}{(k+1)n + 1}$, which gives
				\[
				C_i > \frac{k+1}{k+2}\cdot V_n(\alpha) = \frac{1 - V_n(\alpha)}{n-1},
				\]
				which completes the proof.
			\end{proof}
			
			Interestingly, the following lemma shows the connection between the ranges of $\alpha_i$ and $\frac{\alpha_i}{1 - \frac{1-V_n(\alpha_i)}{n-1}}$. 
			\begin{lemma} \label{lem:n2n-1}
				For any $\alpha_i \in NI(n,k) \cup I (n , k ) $ for some $k \ge 0$,
				we have
				$$
				\frac{\alpha_i}{1 - \frac{1-V_n(\alpha_i)}{n-1}} \in \left\{\begin{array}{ll} 
					I(n - 1, k), & \text { if } \alpha _ i \in I(n,k) \\
					NI(n-1, k), & \text { if } \alpha _ i \in NI(n,k)
				\end{array}\right.
				$$
			\end{lemma}
			\begin{proof}
				We consider the following two cases regarding the ranges of $\alpha_i$.
				
				\medskip

				\textbf{Case 1}: $\alpha_i \in I(n,k)$. In this case, $\frac{k+2}{(k+1)((k+1) n + 1)} \le \alpha_i \le \frac{1}{kn+1}$ and $V_n(\alpha_i) = (k+1)\alpha_i$.
				Then, we have
				\[
				\frac{\alpha_i}{1 - \frac{1-V_n(\alpha_i)}{n-1}} = \frac{(n-1)\alpha_i}{n-2 + (k+1)\alpha_i} \leq \frac{n-1}{(n-2)(kn+1) + k+1} = \frac{1}{k(n-1)+1},
				\]
				where the inequality is because $\alpha_i \le \frac{1}{kn+1}$. 
				Besides, 
				\begin{align*}
					\frac{\alpha_i}{1 - \frac{1-V_n(\alpha_i)}{n-1}} = \frac{(n-1)\alpha_i}{n-2 + (k+1)\alpha_i} &\geq \frac{(k+2)(n-1)}{ (k+1)((n-2)(kn+n+1) + k + 2)} \\
					&= \frac{k+2}{(k+1)((k+1)(n-1) + 1)},
				\end{align*}
				where the inequality is because $\alpha_i \ge \frac{k+2}{(k+1)((k+1) n + 1)}$. 
				
				\medskip

				\textbf{Case 2}: $\alpha_i \in NI(n,k)$. In this case, $\frac{1}{(k+1)n + 1} < \alpha_i < \frac{k+2}{(k+1)((k+1)n + 1)}$ and $V_n(\alpha_i) = \frac{k+2}{(k+1)n + 1}$. 
				Then, we have 
				\[
				\frac{\alpha_i}{1-\frac{1-V_n(\alpha_i)}{n-1}} = \frac{((k+1)n + 1)\alpha_i}{(k+1)(n-1)+1} < \frac{k+2}{(k+1)((k+1)(n-1) + 1)},
				\]
				where the inequality is because $\alpha_i < \frac{k+2}{(k+1)((k+1)n + 1)}$. 
				Besides,
				\[
				\frac{\alpha_i}{1-\frac{1-V_n(\alpha_i)}{n-1}} = \frac{((k+1)n + 1)\alpha_i}{(k+1)(n-1)+1} > \frac{1}{(k+1)(n-1)+1}, 
				\]
				where the inequality is because $\alpha_i > \frac{1}{(k+1)n + 1}$. 
			\end{proof}

			\medskip

			\noindent \begin{proof}[Proof of Theorem \ref{thm:heterogeneous}]
				We prove Theorem \ref{thm:heterogeneous} by mathematical induction.
				When $n = 2$, it is easy to see the correctness of Theorem \ref{thm:heterogeneous} from Lemma \ref{lem:Ci} 
				since $v_1(A_1) \le V_2(\alpha_1)$ and $v_2(A_2) = 1 - v_2(A_1) \le 1 - (1 - V_2(\alpha_2)) = V_2(\alpha_2)$,
				We assume as our induction hypothesis that Theorem \ref{thm:heterogeneous} holds for $n-1$.  
				Then we aim to prove the correctness for $n$. 
				
				From Algorithm \ref{alg:hetero}, $v_1(A_1) \le V_{n}(\alpha_1)$ clearly holds for agent 1. 
				For any other agent $i \in N\setminus \{1\}$, denote $\widetilde{\alpha}_i = \max_{e \in M \setminus A_1}v'_i(e)$. 
				We know from Algorithm \ref{alg:hetero} that $\widetilde{\alpha}_i \leq \frac{\alpha_i}{1 - C_i}$ and from the induction hypothesis that $v'_i(A_i) \le V_{n-1}(\widetilde{\alpha}_i)$, which together give
				\[
				v_i(A_i) = (1-C_i)v'_i(A_i) \le (1-C_i)V_{n-1}(\widetilde{\alpha}_i) \le (1-C_i)V_{n-1}(\frac{\alpha_i}{1 - C_i}),
				\]
				where the last inequality holds by recalling that $V_{n-1}(\widetilde{\alpha}_i)$ is an non-decreasing function of $\widetilde{\alpha}_i$. 
				Therefore, it remains to show 
				\[
				(1-C_i)V_{n-1}(\frac{\alpha_i}{1-C_i}) \leq V_{n}(\alpha_i).
				\]
				
				Note that $(1-C_i)V_{n-1}(\frac{\alpha_i}{1- C_i})$ is an non-increasing function of $C_i$. 
				This is because when $\frac{\alpha_i}{1- C_i} \in I(n-1, k)$ for some $k$, 
				$(1-C_i)V_{n-1}(\frac{\alpha_i}{1- C_i}) = (1-C_i)(k+1)\frac{\alpha_i}{1- C_i} = (k+1)\alpha_i$, which is a constant with respect to $C_i$;
				when $\frac{\alpha_i}{ 1- C_i} \in NI(n-1, k)$ for some $k$, 
				$(1-C_i)V_{n-1}(\frac{\alpha_i}{1- C_i}) = (1-C_i)\frac{k+2}{(k+1)(n-1) + 1}$,
				a decreasing function of $C_i$. 
				It follows that
				\[
				(1-C_i)V_{n-1}(\frac{\alpha_i}{1-C_i}) \le (1-\frac{1-V_{n}(\alpha_i)}{n-1})V_{n-1}(\frac{\alpha_i}{1-\frac{1-V_{n}(\alpha_i)}{n-1}}) = V_n(\alpha_i),
				\]
				where the inequality is due to $C_i \geq \frac{1-V_{n}(\alpha_i)}{n-1}$ (according to Lemma \ref{lem:Ci}),
				and the equality can be verified by considering the following two cases regarding the ranges of $\alpha_i$,
				
				\medskip
				
				\textbf{Case 1}: $\alpha_i \in I(n,k)$. In this case, Lemma \ref{lem:n2n-1} gives $\frac{\alpha_i}{1-\frac{1-V_{n}(\alpha_i)}{n-1}} \in I(n-1, k)$. 
				Thus, we have
				\begin{align*}
					(1-\frac{1-V_{n}(\alpha_i)}{n-1})V_{n-1}(\frac{\alpha_i}{1-\frac{1-V_{n}(\alpha_i)}{n-1}}) &= (1-\frac{1-V_{n}(\alpha_i)}{n-1})\cdot (k+1)\frac{\alpha_i}{1-\frac{1-V_{n}(\alpha_i)}{n-1}} \\
					&= (k+1)\alpha_i = V_n(\alpha_i).
				\end{align*}

				\medskip

				\textbf{Case 2}: $\alpha_i \in NI(n,k)$. In this case, $\frac{\alpha_i}{1-\frac{1-V_{n}(\alpha_i)}{n-1}} \in NI(n-1, k)$.
				Thus, we have
				\begin{align*}
					(1-\frac{1-V_{n}(\alpha_i)}{n-1})V_{n-1}(\frac{\alpha_i}{1-\frac{1-V_{n}(\alpha_i)}{n-1}}) &= (1-\frac{1-\frac{k+2}{(k+1)n + 1}}{n-1})\cdot \frac{k+2}{(k+1)(n-1) + 1} \\
					&= \frac{k+2}{(k+1)n+1} = V_n(\alpha_i). 
				\end{align*}
				Therefore, we complete the proof of Theorem \ref{thm:heterogeneous}.
			\end{proof}
			
			The instances
			provided in Section~\ref{sec::upper-bound-minmaxshare} and Appendix \ref{app:proofs:home}
			show the tightness of
			Theorem~\ref{thm:heterogeneous}.
			
			\section{Numerical Experiments}
			
			\label{sec:experiments}
			To demonstrate that $\Delta _{n}^{\oplus}(\alpha; m)$ can serve as a good alternative of MinMaxShare, we first evaluate the worst-case ratio of $\Delta _{n}^{\oplus}(\alpha; m)$ and $\Delta _{n}^{\circleddash}(\alpha; m)$ (recall that $\Delta _{n}^{\circleddash}(\alpha; m)$ is the best-case MinMaxShare over all disutilities in $\mathcal{V}(\alpha; m)$). 
			Denote by $r_n(\alpha; m) = \frac{\Delta _{n}^{\oplus}(\alpha; m)}{\Delta _{n}^{\circleddash}(\alpha; m)}$. 
			It is clear that $r_n(\alpha; m)$ is no smaller than the ratio between $\Delta _{n}^{\oplus}(\alpha; m)$ and the real MinMaxShare, and we have illustrated $r_n(\alpha; \infty)$ in Fig. \ref{fig:ratio:theoretical} for $n = 2, 10, 100$.
			As we can see, although the worst-case ratio can be close to 2, it only happens for sufficiently large $n$ and a small range of values of $\alpha$.
			For any $n$ and most values of $\alpha$, the ratio is better than $\frac{4}{3}$ and $\frac{11}{9}$, which are two fractions of the
			MinMaxShare that are known to be achievable. 
			Actually, it is not hard to verify that $r_n(\alpha; m) 
			\le \frac{2n}{n+1} < 2$ for all $\alpha$, and $r_n(\alpha; m) \le \frac{4}{3}$ for all $\alpha$ outside of $(\frac{4}{9n}, \frac{3}{2n+3})$; we provide simple proofs in the appendix. 
			Note that $\frac{3}{2n+3} - \frac{4}{9n} < \frac{7}{6n}$.
			
			\begin{claim}
				\label{claim:r}
				For any $n \ge 2$, $\alpha \in (0, 1]$ and $m \ge \lceil \frac{1}{\alpha} \rceil$, $r_n(\alpha; m) \le \frac{2n}{n+1}$. 
			\end{claim}
			
			\begin{claim}
				\label{claim:large}
				$r_n(\alpha; m) > \frac{4}{3}$ only when $\alpha \in (\frac{2}{9}, \frac{1}{3})$ if $n=3$, or $\alpha \in (\frac{1}{6}, \frac{3}{11})$ if $n=4$, or $\alpha \in (\frac{4}{45}, \frac{1}{9}) \cup (\frac{2}{15}, \frac{3}{13})$ if $n=5$, or $\alpha \in (\frac{4}{9n}, \frac{3}{2n+3})$ if $n\ge 6$.
			\end{claim}
			
			From the formula of $r_n(\alpha; m)$, as well as Fig. \ref{fig:ratio:theoretical}, we have the following observations: 
			\begin{description}
				\item[Observation 1] As $n$ increases, the worst-case ratio of $r_n(\alpha; m)$, i.e., $\max_{\alpha} r_n(\alpha; m)$, increases.
				\item[Observation 2] As $n$ increases, large values of $r_n(\alpha; m)$ happen increasingly more rarely if $\alpha$ is randomly generated from $[0,1]$.
			\end{description}
			
			Next, we conduct numerical experiments with synthetic and real-world data to illustrate the real distances between $\Delta _{n}^{\oplus }(\alpha; m)$ and the MinMaxShare of specific disutility functions, which also validate the above two observations.

			\subsection{Experiments with Synthetic Data}

			In this section, we randomly generate a number of disutility functions, and for each of them, we compute the ratio between the corresponding Hill's share and the MinMaxShare. 
			In particular, 
			for each given $n$ and $m$, we randomly generate 100 instances; for each instance, we randomly generate $m-1$ numbers in $[0, 1]$. 
			These $m-1$ numbers separate the interval $[0, 1]$ into $m$ contiguous segments, and the lengths of these segments are used as the disutilities of the $m$ objects. 
			Then we compute the $\Delta_n^{\oplus}(\alpha; m)$ value using the maximum of these values and 
			the MinMaxShare.
			For each instance, we record the ratio of these two quantities.
			
			The results are summarized in Fig. \ref{fig:ratio:random}. 
			We slice the ratios into small intervals, each of which has a length of
			0.1, and count the number of instances falling into each interval for each setting. 
			The figure validates the previous two observations:
			when $n=2$ and $3$, the largest ratio can only reach interval $[1.3, 1.4)$ and $[1.4, 1.5)$, but when $n\ge 4$, it reaches $[1.5, 1.6)$;
			however, looking at the number of instances, for larger $n$, fewer and fewer instances fall into these large intervals, and instead, the number of instances in $[1.0,1.1)$ significantly dominates the other intervals. 
			Specifically, when $n=6$ and $7$, $[1.0,1.1)$ contains over $80\%$ of all random instances, and none of them reaches a ratio beyond 1.6,
			while the worst-case ratio can be greater than 1.7.
			
			\begin{figure}
				\centering
				\subfigure[$n=2$, $m=8, 9, 10$]{
					\includegraphics[width=0.45\linewidth]{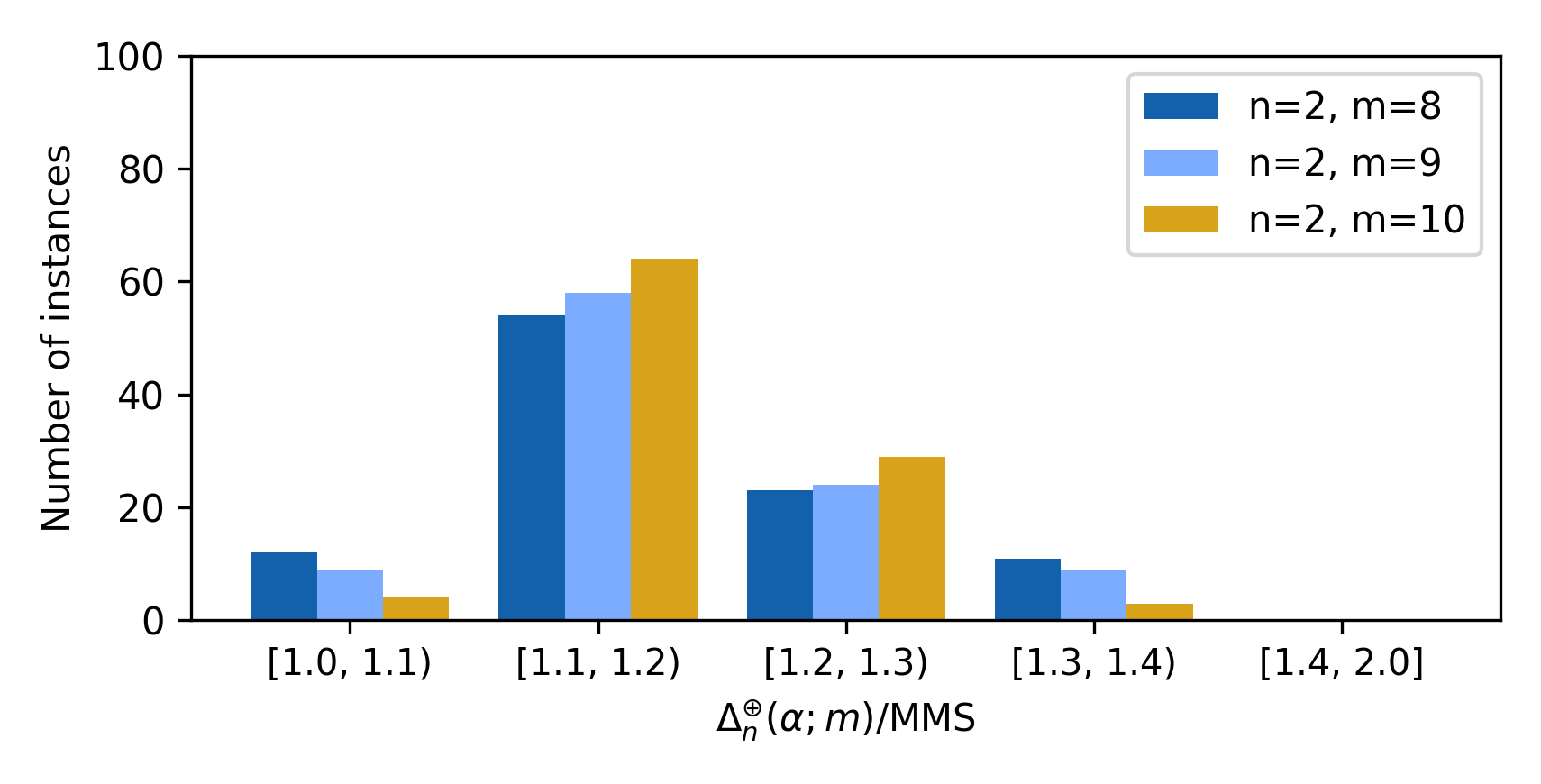}
				}
				\quad
				\subfigure[$n=3$, $m=8, 9, 10$]{
					\includegraphics[width=0.45\linewidth]{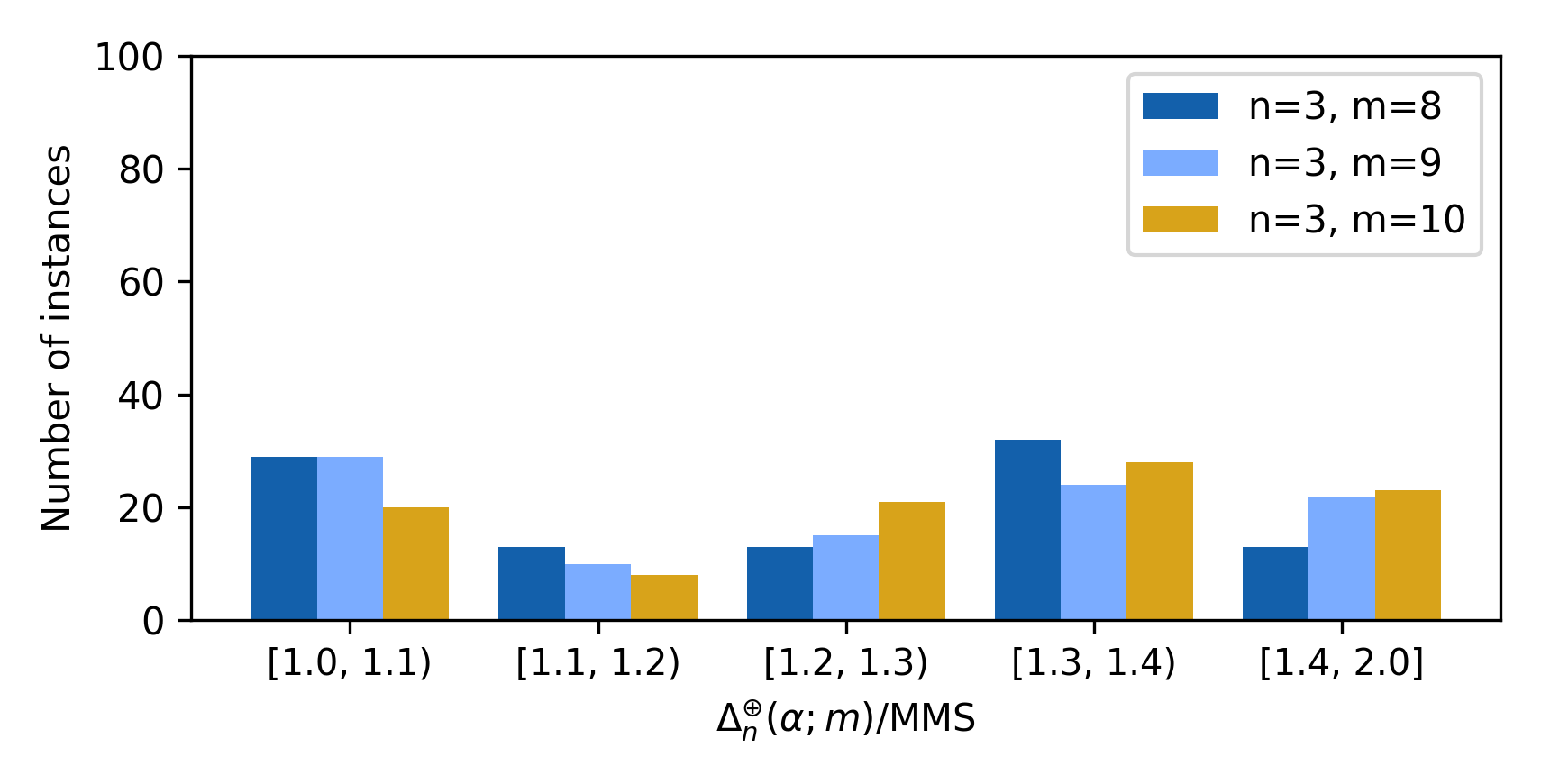}
				}
				\quad
				\subfigure[$n=4$, $m=8, 9, 10$]{
					\includegraphics[width=0.45\linewidth]{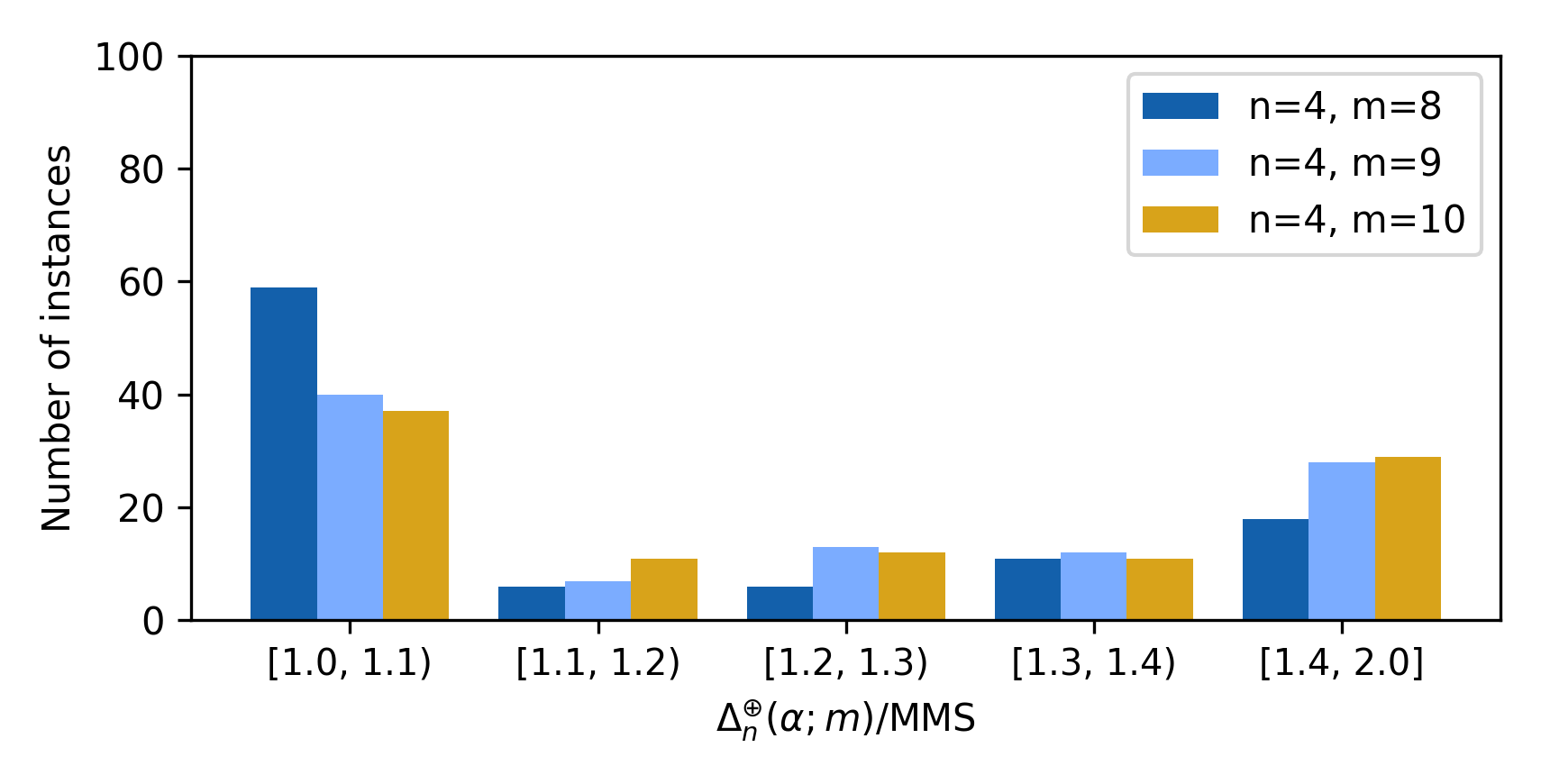}
				}
				\quad
				\subfigure[$n=5$, $m=8, 9, 10$]{
					\includegraphics[width=0.45\linewidth]{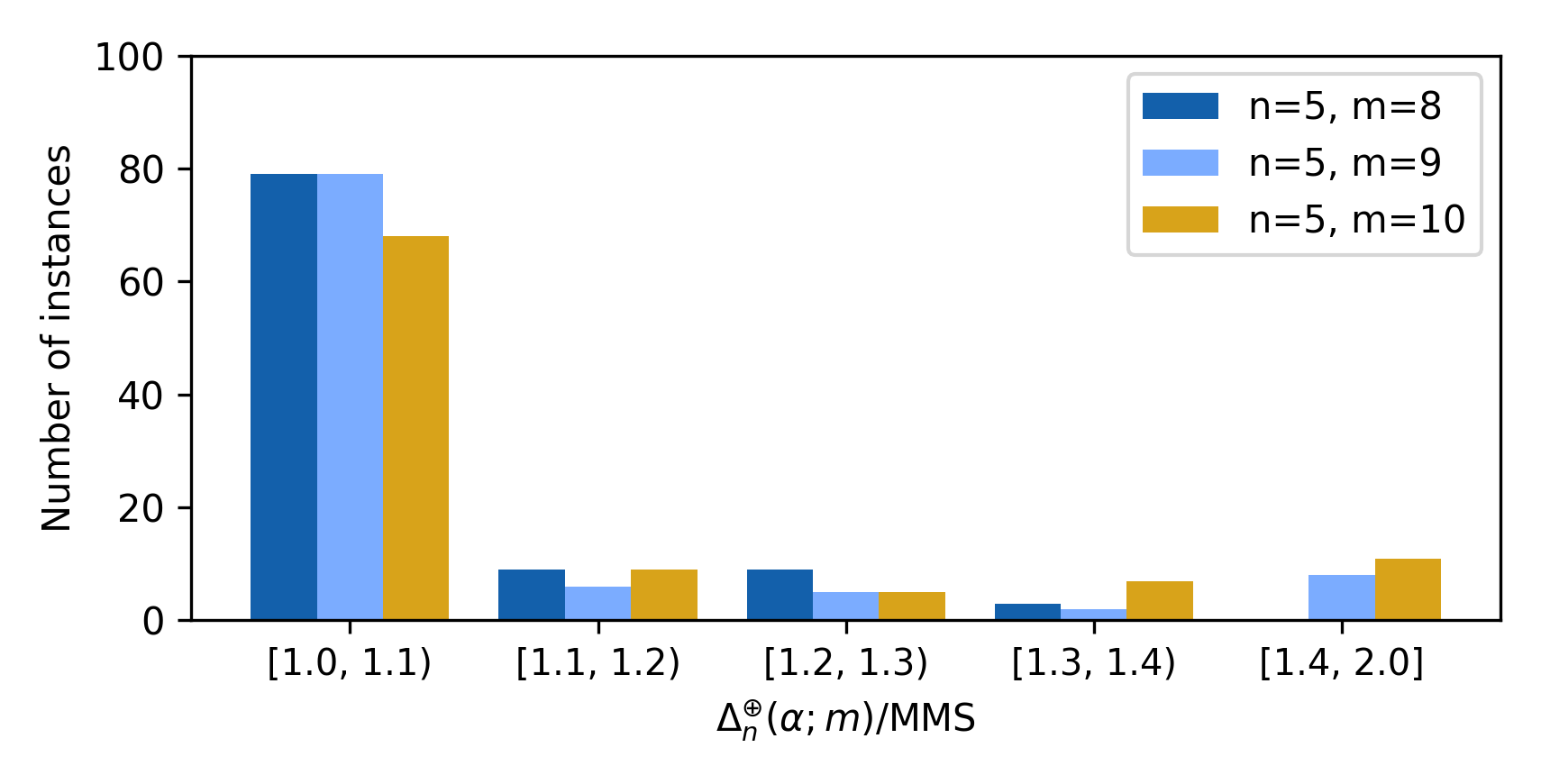}
				}
				\quad
				\subfigure[$n=6$, $m=8, 9, 10$]{
					\includegraphics[width=0.45\linewidth]{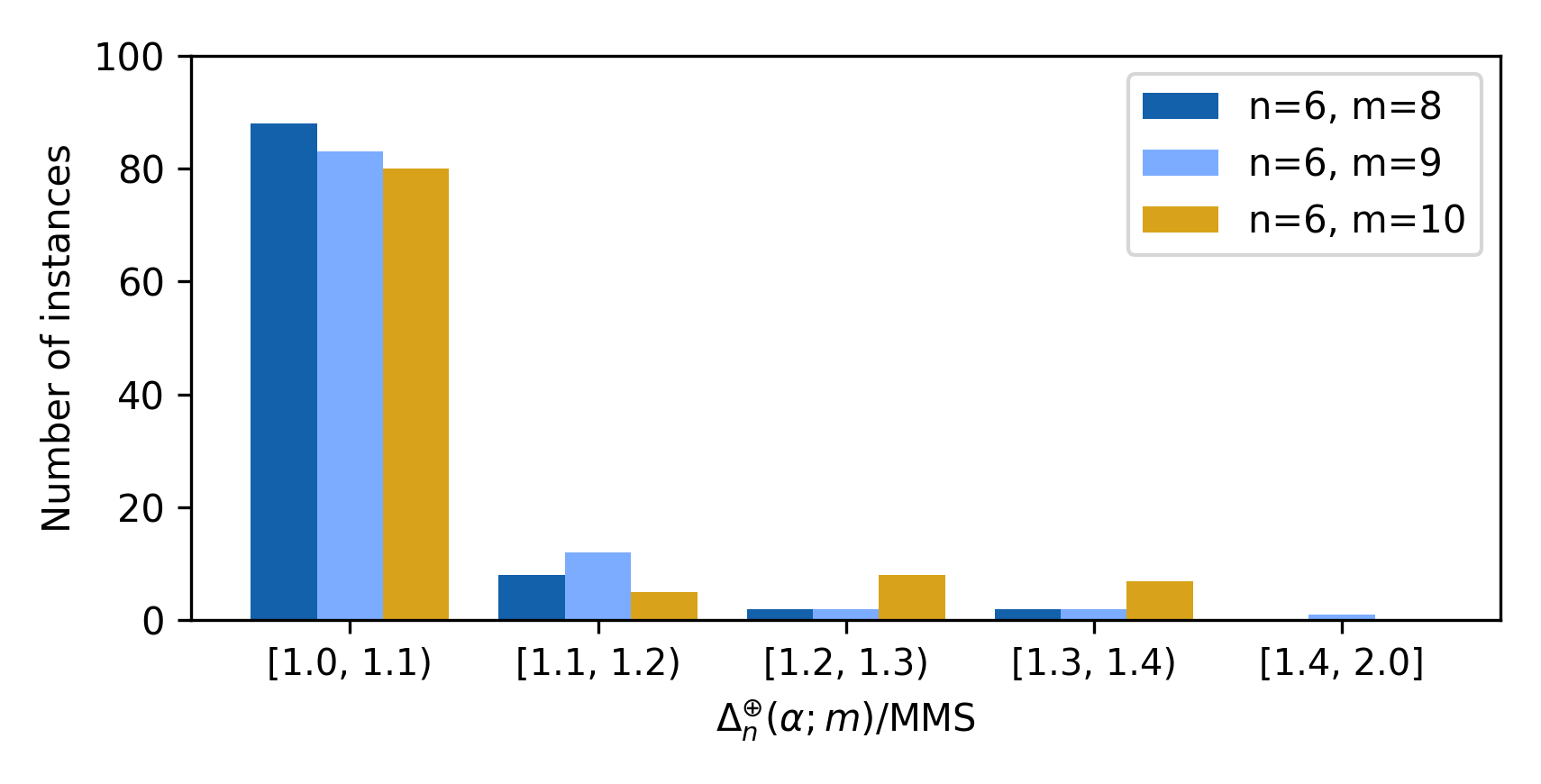}
				}
				\quad
				\subfigure[$n=7$, $m=8, 9, 10$]{
					\includegraphics[width=0.45\linewidth]{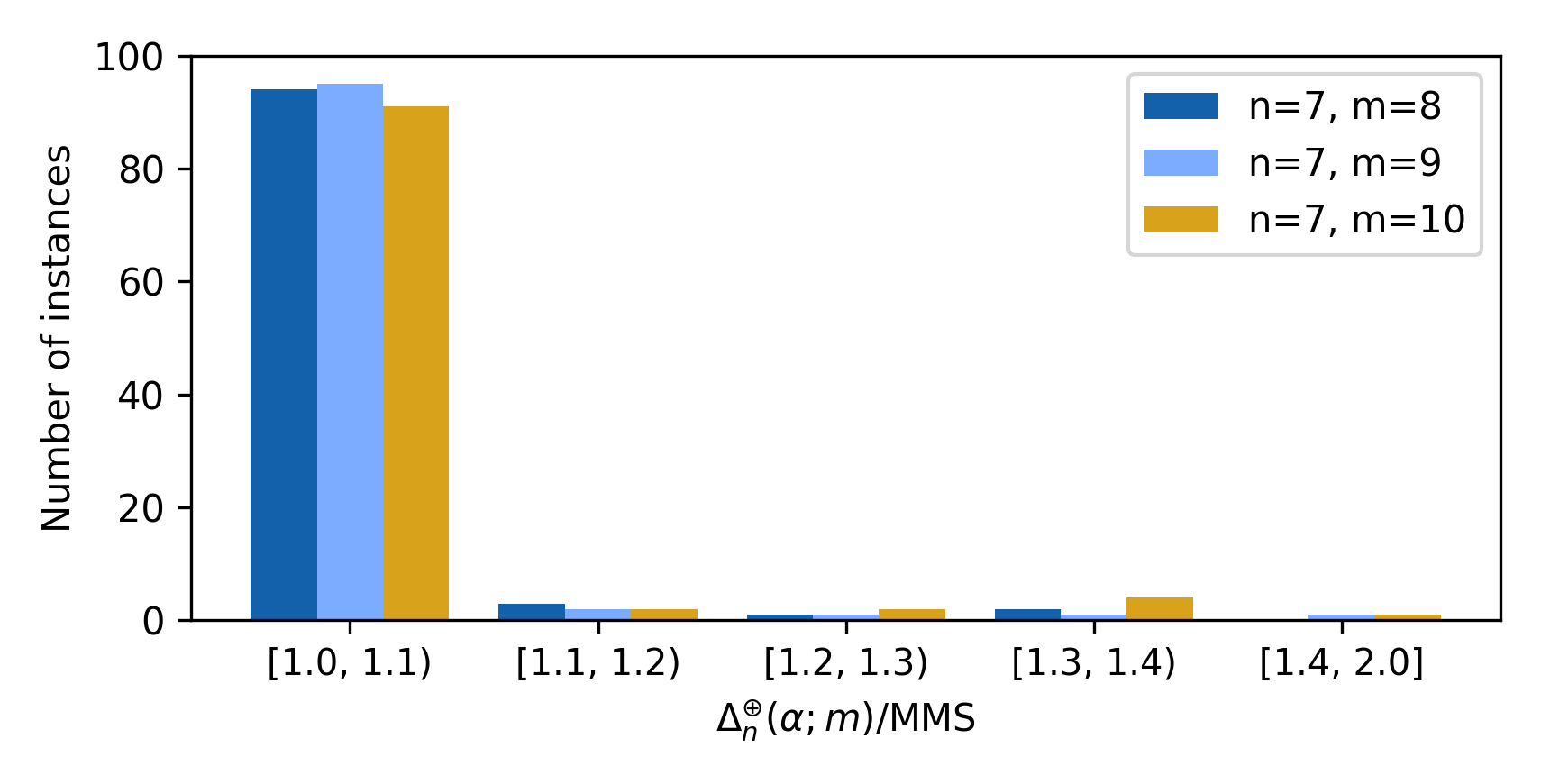}
				}
				\caption{Ratios in random data.}
				\label{fig:ratio:random}
			\end{figure}
			
			In the appendix, we conduct more experiments by fixing $n=2$ and increasing the value of $m$ and report the change in the distribution of the ratios.

			\subsection{Experiments with Real-World Data}
			The real-world data set is collected from the Spliddit platform (spliddit.org) -- a well-known platform that provides implementations of fair allocation algorithms for various practical problems \citep{DBLP:journals/sigecom/GoldmanP14}. 
			The data set contains 8,409 instances created between October 2014 and May 2020, involving 22,530 agents and 42,469 objects. 
			We randomly select 10,000 disutility functions from the data, where the largest value of $n$ is $14$. After normalising all the disutility functions, 
			for each of them, we record the ratio of the corresponding Hill's share and the MinMaxShare. 
			The results are shown in Fig. \ref{fig:ratio:spliddit}.
			As we can see, very few instances have ratios higher than 1.4, and over $65\%$ of the instances have ratios within $[1.0,1.1)$.
			Actually, there are only 173 ($= 1.73\%$) and 26 ($= 0.26\%$) instances falling into $[1.6,1.7)$ and $[1.7, 1.8)$ respectively, and none is beyond 1.8.
			Note that in the 10,000 disutility functions, there are only 14 instances with $n\ge 9$, which further amplifies the rare happening of large ratios.

			\begin{figure}[H]
				\centering
				\includegraphics[width=0.85\textwidth]{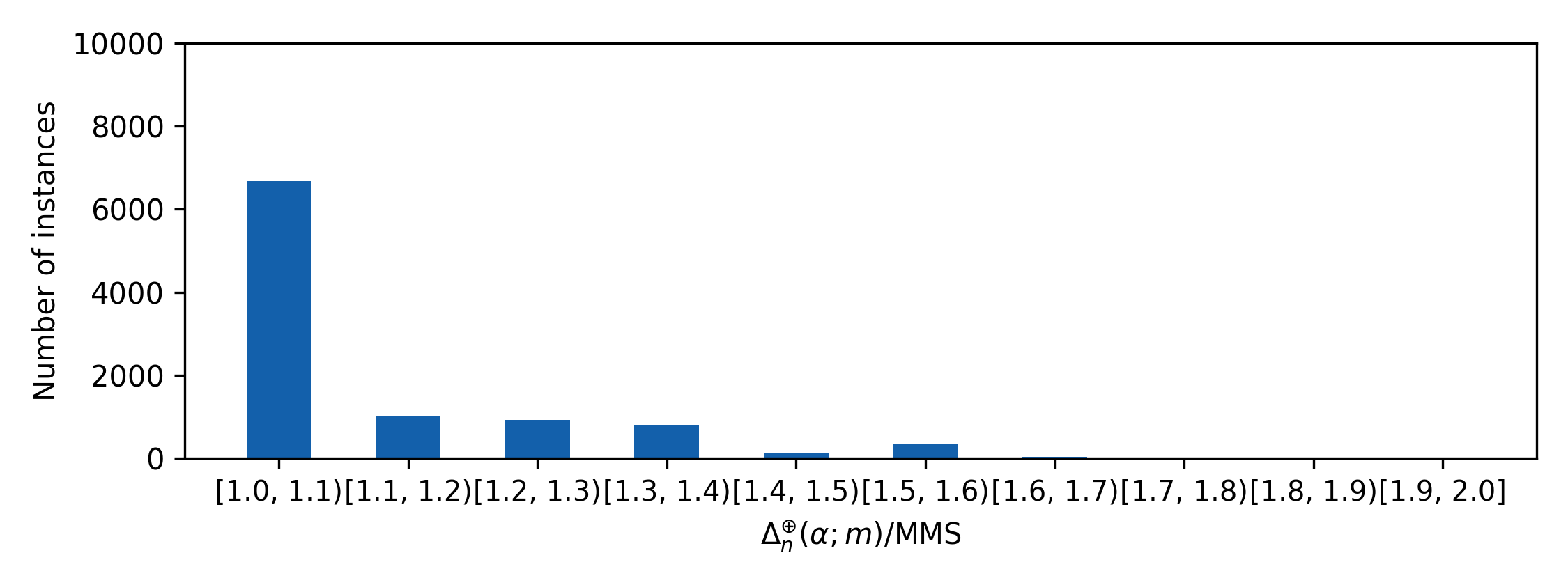}
				\caption{Ratios in Spliddit data.}
				\label{fig:ratio:spliddit}
			\end{figure}
			

			\section{Conclusion}
			
			In this work, we give the tight characterisation of Hill's share for allocating indivisible bads, i.e., the exact upper bound of the MinMaxShare of disutility functions with the same largest single-object value. 
			Hill's share exhibits several advantages including elementary computation, being close to the MinMaxShare, and displaying the effect of an agent's disutility in her share of all objects. 
			More importantly, the monotonic cover of Hill's share serves as a canonical guarantee; as far as we know, no other similarly simple guarantee for the allocation of bads has been identified.
			There are some open problems. 
			Hill's guarantee is tight for the domain of disutility functions whose largest single-object disutility is {\em no greater than} a given parameter, but we do not know whether it is tight when the domain only contains the disutility functions whose largest single-object disutility {\em equals} this parameter.
			The same problem is also open for the mirror problem of allocating goods, for which the tight characterisation of Hill's share is also unknown (when $n \ge 3$).
			Our work also uncovers some other related research problems, such as the algorithmic problem of finding a Pareto optimal allocation satisfying Hill's share and the game-theoretic problem of designing truthful mechanisms to incentivize the agents to report their disutility functions honestly while achieving (approximations of) Hill's share.

			\bibliography{EC_mylibrary}

			\newpage
			\appendix
			\section*{Appendix}
			
			\section{Missing Proofs in Section \ref{sec:pre}}
			\subsection{Proof of Lemma \ref{lem:bestMMS_m}}
			For each case, we show that $\MMS_{n}(v) \ge \Delta _{n}^{\circleddash}(\alpha; m)$ for any $v \in \mathcal{V}(\alpha ;m)$, 
			and design a disutility function such that the
			MinMaxShare is exactly $\Delta _{n}^{\circleddash}(\alpha; m)$. 
			By the definition of $\mathcal{V}(\alpha ;m)$, there exists an object with disutility $\alpha$, thus $\MMS_{n}(v) \ge \alpha$ for any $v \in \mathcal{V}(\alpha ;m)$. 
			Moreover, when $\alpha > 1/n$, there exists a disutility function such that the MinMaxShare is exactly $\alpha$. 
			Specifically, $v_1$ contains $\lceil \frac{1}{\alpha} \rceil$ objects, 
			$\lfloor \frac{1}{\alpha} \rfloor$ with disutility $\alpha$ 
			and one with disutility $(1 - \lfloor \frac{1}{\alpha}\rfloor \cdot \alpha) < \alpha$ (if 1 is indivisible by $\alpha$). 
			$\MMS_{n}(v_1) = \alpha$ follows from the fact that $v_1$ contains at most $n$ objects. 
			
			By the definition of MinMaxShare, $\MMS_{n}(v) \ge \frac{1}{n}$, where the equality is achieved when the total disutility of $M$ can be evenly distributed among the $n$-partition. 
			When $1/n$ is divisible by $\alpha$ (i.e., $\alpha = \frac{1}{kn}$ for some positive integer $k$), 
			or $1/n$ is not divisible by $\alpha$ 
			(i.e., $\frac{1}{(k+1)n} < \alpha < \frac{1}{kn}$) and the number of objects $m$ is at least $kn+n$, 
			there exists an disutility function such that the MinMaxShare is exactly $1/n$. 
			For the former, the disutility function $v_2$ contains $1/\alpha = kn$ objects with disutility $\alpha$. 
			Clearly, each bundle in the best $n$-partition contains $k$ objects with disutility $\alpha$ and $\MMS_{n}(v_2) = 1/n$. 
			For the latter, intuitively, the total disutility of $M$ can also be evenly distributed by letting each bundle contain $\lfloor \frac{1}{n\alpha} \rfloor = k$ objects with disutility $\alpha$ and one object with disutility $\frac{1}{n} - k\alpha < \alpha$. 
			In total, $kn+n \le m$ objects are needed.  
			Hence, the disutility function $v_3$ contains $kn$ objects with disutility $\alpha$, $n$ objects with disutility $\frac{1}{n} - k\alpha$ and $m - kn - n$ objects with disutility 0, and $\MMS_{n}(v_3) = 1/n$. 
			
			However, when $1/n$ is indivisible by $\alpha$ but the number of objects $m$ is limited to $kn+n-1$,
			$1/n$ cannot be achieved since some bundles in any $n$-partition contain no more than $k$ objects,
			and the disutilities of these bundles are at most $k\alpha < 1/n$. 
			For this case, we show that $\MMS_{n}(v) \ge k\alpha + \frac{1-kn\alpha}{m-kn}$ for any $v \in \mathcal{V}(\alpha ;m)$. 
			Let $x$ be the number of bundles in the $n$-partition that contain no more than $k$ objects, it follows that $x \ge kn+n-m$. 
			Since the disutility of each of these bundles is at most $k\alpha$, the average disutility of the other bundles is at least 
			\[
			\frac{1-k\alpha x}{n-x} \ge \frac{1-(kn+n-m)\cdot k\alpha}{m-kn} = k\alpha + \frac{1-kn\alpha}{m-kn} > k\alpha
			\]
			where the leftmost-hand side is an increasing function of $x$ since $k\alpha < 1/n$, and the last inequality is because $m \ge \lceil \frac{1}{\alpha} \rceil > kn$. 
			Therefore, the largest disutility of any $n$-partition is at least $k\alpha + \frac{1-kn\alpha}{m-kn}$; that is, $\MMS_{n}(v) \ge k\alpha + \frac{1-kn\alpha}{m-kn}$ for any $v \in \mathcal{V}(\alpha ;m)$. 
			Let $v_4$ contain $kn$ objects with disutility $\alpha$ and $m-kn$ objects with disutility $\frac{1-kn\alpha}{m-kn} < \alpha$.  
			Clearly, the worst bundle of the best $n$-partition contains $k$ objects with disutility $\alpha$ and one object with disutility $\frac{1-kn\alpha}{m-kn}$, thus $\MMS_{n}(v_4) = k\alpha + \frac{1-kn\alpha}{m-kn}$.
			
			\subsection{Proof of Lemma \ref{lem:worstMMS_property}}
			That $\Delta _{n}^{\oplus }(\alpha ;m)$
			decreases in $n$ is clear by comparing the MinMaxShares of an arbitrary $n$-partition and the $(n+1)$-partition obtained by adding one empty share.
			The monotonicity in $m$ (i.e., $\Delta _{n}^{\oplus }(\alpha ;m) \le \Delta _{n}^{\oplus }(\alpha ;m+1)$) follows that every disutility in $\mathcal{V}(\alpha; m)$ can be transformed to one in $\mathcal{V}(\alpha; m+1)$ by adding an object with disutility 0, without changing the MinMaxShare. 
			
			We then show when $m \ge \lceil \frac{2}{\alpha} \rceil - 1$, $\Delta _{n}^{\oplus }(\alpha ;m) \ge \Delta _{n}^{\oplus }(\alpha ;m+1)$, thus $\Delta _{n}^{\oplus }(\alpha ;m)$ remains constant. 
			To achieve this, we first claim that when $m \ge \lceil \frac{2}{\alpha} \rceil - 1$, for any $v \in \mathcal{V}(\alpha ;m+1)$ and any allocation $(A_1, \ldots, A_n)$, there exists one bundle such that the total disutility of two of its objects is no more than $\alpha$. 
			Otherwise, for any bundle $A_k$, the total disutility of any two objects is larger than $\alpha$, which means that $v(A_k) > \frac{|A_k|}{2}\alpha$. 
			Upon summing up the lower bounds over all bundles, $1 = \sum_{k\in N}v(A_k) > \frac{\alpha}{2}\cdot \lceil \frac{2}{\alpha} \rceil \ge 1$, a contradiction. 
			
			Now we pick any disutility $v \in \mathcal{V}(\alpha; m+1)$, and let $(A_1, \ldots, A_n)$ be the allocation that gives the MinMaxShare of $v$. 
			By the claim, there exists a bundle (w.l.o.g., $A_1$) such that two objects $e_1, e_2 \in A_1$ satisfy $v(e_1) + v(e_2) \le \alpha$. 
			We derive a disutility $v' \in \mathcal{V}(\alpha; m)$ by merging $e_1$ and $e_2$ into one object $e$, and show that $\MMS_n(v) = \MMS_n(v')$. 
			On one hand, let $A'_1 = A_1 \setminus \{e_1, e_2\} \cup \{e\}$, since $(A'_1, \ldots, A_n)$ is an allocation regarding $v'$ with the largest disutility being $\MMS_n(v)$, it follows that $\MMS_n(v) \ge \MMS_n(v')$. 
			On the other hand, by decomposing $e$ into $e_1$ and $e_2$, we can convert any allocation regarding $v'$ to an allocation regarding $v$ without changing the largest disutility, thus $\MMS_n(v) \le \MMS_n(v')$. 
			
			Therefore, when $m \ge \lceil \frac{2}{\alpha} \rceil - 1$, every disutility in $\mathcal{V}(\alpha; m+1)$ can be transformed to one in $\mathcal{V}(\alpha; m)$ without changing the MinMaxShare, which gives $\Delta _{n}^{\oplus }(\alpha ;m) \ge \Delta _{n}^{\oplus }(\alpha ;m+1)$. 
			By combining the monotonicity in $m$, $\Delta _{n}^{\oplus }(\alpha ;m)$ remains constant when $m \ge \lceil \frac{2}{\alpha} \rceil - 1$. 
			
			\section{Missing Proofs in Section ~\ref{sec::upper-bound-minmaxshare}}
			\label{app:proofs:home}

			\subsection{Case 3: $n = 2$ and $k = 1$ for unrestricted $m$}
			\label{subsubsec:unrestricted:n=2-k=1}
			
			We now prove Corollary \ref{coro:homogeneous} for the case of $n = 2$ and $k = 1$, i.e., $\alpha \in D(2,1)\cup I(2,1)$.
			
			\medskip
			
			\textbf{Subcase 3.1: $\alpha \in (\frac{1}{5}, \frac{7}{27}]$} 
			
			\medskip
			
			When $\alpha \in (\frac{1}{5}, \frac{7}{27}]$, $v(A_1) > \Delta_2^{\oplus}(\alpha) = \frac{3-3\alpha}{4}$. 
			If $E_{\alpha} \cap A_2 \neq \emptyset$, $A_2$ contains some objects with disutility $\alpha$. 
			Recall that Claim \ref{claim:subset} holds as long as $k = 1$, thus there exists $S \subseteq A_1$ such that $\frac{1}{3}v(A_1) \le v(S) < v(A_1) - \alpha$. 
			Denote by $e^*$ one object in $A_2$ with disutility $\alpha$, Claim \ref{clm:v_Aj} gives $v(A_2 \setminus \{e^*\}) \ge v(S) \ge \frac{1}{3}v(A_1)$. 
			As a result, 
			\[
			1 = v(A_1) + v(A_2) \ge v(A_1) + \frac{1}{3}v(A_1) + \alpha,
			\]
			which gives $v(A_1) \le \frac{3-3\alpha}{4}$, thus contradicting the assumption that $v(A_1) > \Delta_2^{\oplus}(\alpha)$. 
			
			Therefore, $E_{\alpha} \cap A_2 = \emptyset$, 
			which means that all the objects with disutility $\alpha$ are in $A_1$ and for any $e \in A_2$, $v(e) < \alpha$. 
			We first derive an upper bound and a lower bound of the maximum disutility of the objects in $A_2$. 
			Denote by $e^*$ one object in $A_1$ with $v(e^*)=\alpha < v(A_1)$, since $v(A_1) - v(A_2) = 2v(A_1) - 1 > \frac{1 - 3\alpha}{2}$, Claim \ref{clm:diff} gives
			\[
			\max_{e\in A_2}v(e) \le v(e^*) - (v(A_1) - v(A_2)) < \frac{5\alpha-1}{2}. 
			\]
			Notice that $\frac{1-3\alpha}{2} > \frac{\alpha}{3}$ since $\alpha \le \frac{7}{27} < \frac{3}{11}$, $v(A_1) - v(A_2) > \frac{\alpha}{3}$. 
			Then for every $S \subseteq A_2$ with $v(S) < \alpha$, Claim \ref{clm:diff} actually gives a tighter bound of $v(S)$, i.e., $v(S) \le v(e^*) - (v(A_1) - v(A_2)) < \frac{2}{3}\alpha$. 
			This also implies that for every $S' \subseteq A_2$ with $v(S') \ge \frac{2}{3}\alpha$, $v(S') \ge \alpha$ actually holds. 
			Let $S^* = \arg\min_{S\subseteq A_2, v(S) \ge \frac{2}{3}\alpha}v(S)$ whose existence is guaranteed since Claim \ref{clm:v_Aj} gives $v(A_2) \ge v(e^*) = \alpha$, thus, $v(S^*) \ge \alpha$. 
			Then from the definition of $S^*$, $v(e) \ge v(S^*) - \frac{2}{3}\alpha \ge \frac{1}{3}\alpha$ holds for any $e \in A_2$,
			which implies 
			\[
			\max_{e \in A_2}v(e) \ge \frac{\alpha}{2}.
			\]
			Otherwise (i.e., $\max_{e \in A_2}v(e) < \frac{\alpha}{2}$), the total disutility of any two objects in $A_2$ is at least $\frac{2}{3} \alpha$ and smaller than $\alpha$, which is a contradiction to Claim \ref{clm:diff}. 
			
			We then show that $|A_1|$ is exactly 3. 
			Otherwise (i.e., $|A_1| \ge 4$), there exists $S \subseteq A_1$ such that $v(A_1) > v(S) \ge \alpha + \frac{2}{3}(v(A_1)-\alpha)$.
			Then Claim \ref{clm:v_Aj} gives $v(A_2) \ge v(S) \ge \alpha + \frac{2}{3}(v(A_1)-\alpha)$. 
			Summing up the lower bounds of $v(A_1)$ and $v(A_2)$ leads to a contradiction as below
			\[
			1 = v(A_1) + v(A_2) \ge \frac{5}{3}v(A_1) + \frac{1}{3}\alpha > \frac{15-11\alpha}{12} > 1,
			\]
			where the last inequality is because $\alpha \le \frac{7}{27} < \frac{3}{11}$. 
			Therefore, we can denote $A_1 = \{e_1^1, e_2^1, e_3^1\}$ and assume without loss of generality that $v(e_1^1)=\alpha \ge v(e_2^1)=x \ge v(e_3^1)=y$. 
			We then derive the lower bounds of $x$ and $y$, and reveal the contradiction in this subcase. 
			Since $x \ge y$, the following formula holds,
			\[
			x \ge \frac{x+y}{2} = \frac{v(A_1) - \alpha}{2} > \frac{3 - 7\alpha}{8} \ge \frac{5\alpha-1}{2} > \max_{e\in A_2}v(e),
			\]
			where the second last inequality is because $\alpha \le \frac{7}{27}$. 
			Then Claim \ref{clm:diff} gives the following lower bound of $x$,
			\[
			x \ge \max_{e\in A_2}v(e) + (v(A_1) - v(A_2)) > \frac{\alpha}{2} + \frac{1-3\alpha}{2} = \frac{1-2\alpha}{2}. 
			\]
			Claim \ref{clm:diff} also gives $y \ge v(A_1) - v(A_2)$. 
			Notice that 
			\[
			2 \cdot (v(A_1) - v(A_2)) > \frac{2 - 6\alpha}{2} > \alpha - \frac{1 - 3\alpha}{2} > x - (v(A_1) - v(A_2)) ,
			\]
			where the second inequality is because $\alpha \le \frac{7}{27} < \frac{3}{11}$, we have the following lower bound of $y$
			\[
			y > \frac{1}{2} \cdot (x - (v(A_1) - v(A_2))) \ge \frac{1}{2} \cdot \max_{e\in A_2}v(e) \ge \frac{\alpha}{4}. 
			\]
			Therefore, $v(A_1) = \alpha + x + y > \alpha + \frac{1-2\alpha}{2} + \frac{\alpha}{4} = \frac{2+\alpha}{4}$, which gives $v(A_1) - v(A_2) = 2v(A_1)-1 > \frac{\alpha}{2}$. 
			However, according to Claim \ref{clm:diff}, $v(A_1) - v(A_2) \le \alpha - \max_{e\in A_2}v(e) \le \frac{\alpha}{2}$, thus constituting a contradiction. 
			
			For the other direction, the disutility function for this subcase contains one object with disutility $\alpha$ and four objects with disutility $\frac{1-\alpha}{4}$. 
			Since $\frac{1}{5} < \alpha \le \frac{7}{27}$, it follows that $\frac{1-\alpha}{4} < \alpha < 2\cdot \frac{1-\alpha}{4}$, where the last inequality is because $\alpha \le \frac{7}{27} < \frac{1}{3}$. 
			Clearly, the MinMaxShare of this disutility function is $3\cdot \frac{1-\alpha}{4}$.
			
			\medskip
			
			\textbf{Subcase 3.2: $\alpha \in (\frac{7}{27}, \frac{2}{7}]$} 
			
			\medskip
			
			When $\alpha \in (\frac{7}{27}, \frac{2}{7}]$, $v(A_1) > \Delta_2^{\oplus}(\alpha) = \frac{2+3\alpha}{5}$. 
			If $E_\alpha \cap A_2 \neq \emptyset$, the proof is similar to that for the counterpart of Subcase 3.1. 
			That is, we also have $v(A_1) \le \frac{3-3\alpha}{4}$, which contradicts $v(A_1) > \Delta_n^{\oplus}(\alpha)$ since $\frac{3-3\alpha}{4} < \frac{2+3\alpha}{5}$ when $\alpha > \frac{7}{27}$. 
			
			Therefore, we can focus on $E_\alpha \cap A_2 = \emptyset$. 
			We first derive an upper bound and a lower bound of the maximum disutility of the objects in $A_2$, 
			which is similar to the counterpart of Subcase 3.1.
			Denote by $e^*$ one object in $A_1$ with $v(e^*)=\alpha < v(A_1)$, since $v(A_1) - v(A_2) = 2v(A_1) - 1 > \frac{6\alpha - 1}{5}$, Claim \ref{clm:diff} gives
			\[
			\max_{e\in A_2}v(e) \le v(e^*) - (v(A_1) - v(A_2)) < \frac{1-\alpha}{5}. 
			\]
			Notice that $\frac{6\alpha - 1}{5} > \frac{\alpha}{3}$ since $\alpha > \frac{7}{27} > \frac{3}{13}$, $v(A_1) - v(A_2) > \frac{\alpha}{3}$. 
			Then for every $S \subseteq A_2$ with $v(S) < \alpha$, Claim \ref{clm:diff} actually gives a tighter bound of $v(S)$, i.e., $v(S) \le v(e^*) - v(A_1) - v(A_2) < \frac{2}{3}\alpha$. 
			This also implies that for every $S' \subseteq A_2$ with $v(S') \ge \frac{2}{3}\alpha$, $v(S') \ge \alpha$ actually holds. 
			Let $S^* = \arg\min_{S\subseteq A_2, v(S) \ge \frac{2}{3}\alpha}v(S)$ whose existence is guaranteed since Claim \ref{clm:v_Aj} gives $v(A_2) \ge v(e^*) = \alpha$, thus, $v(S^*) \ge \alpha$. 
			Then from the definition of $S^*$, $v(e) \ge v(S^*) - \frac{2}{3}\alpha \ge \frac{1}{3}\alpha$ holds for any $e \in A_2$,
			which implies 
			\[
			\max_{e \in A_2}v(e) \ge \frac{\alpha}{2}.
			\]
			Otherwise (i.e., $\max_{e \in A_2}v(e) < \frac{\alpha}{2}$), the total disutility of any two objects in $A_2$ is at least $\frac{2}{3} \alpha$ and smaller than $\alpha$, which is a contradiction to Claim \ref{clm:diff}. 
			
			Observe that $A_1$ contains exactly one object with disutility $\alpha$. 
			Otherwise (i.e., $A_1$ contains at least two objects with disutility $\alpha$), Claim \ref{clm:v_Aj} gives $v(A_2) \ge 2\alpha$ which leads to the following contradiction
			\[
			1 = v(A_1) + v(A_2) > \frac{2+3\alpha}{5} + 2\alpha > 1,
			\]
			where the last inequality is because $\alpha > \frac{7}{27} > \frac{3}{13}$. 
			Recall that $|A_1| \ge 3$, $A_1$ contains at least two objects with disutility smaller than $\alpha$. 
			For each of such objects, we call it a \textit{medium object} if its disutility is larger than $\max_{e\in A_2}v(e)$. Otherwise, we call it a \textit{small object}. 
			Then Claim \ref{clm:diff} gives the following lower bound of the disutility of any medium object $e$
			\begin{align*}
				v(e) &\ge \max_{e\in A_2}v(e) + (v(A_1) - v(A_2)) \\
				&= \max_{e\in A_2}v(e) + (2v(A_1) - 1) > \frac{\alpha}{2} + \frac{6\alpha-1}{5} = \frac{17\alpha-2}{10}, 
			\end{align*}
			as well as the following lower bound of the disutility of any small object $e'$
			\[
			v(e') \ge v(A_1) - v(A_2) = 2v(A_1) - 1 > \frac{6\alpha-1}{5}. 
			\]
			
			We then reveal the contradiction by considering possible combinations of objects in $A_1$ and showing that no possible combination exists. 
			
			\medskip
			
			\textit{Combination 1}: besides the object with disutility $\alpha$, $A_1$ also contains at least 3 small objects. 
			Thus, $v(A_1) > \alpha + 3\cdot \frac{6\alpha-1}{5} = \frac{23\alpha-3}{5}$. 
			Then a lower bound of the difference between $v(A_1)$ and $v(A_2)$ is 
			\[
			v(A_1) - v(A_2) = 2v(A_1) - 1 > \frac{46\alpha - 11}{5} > \frac{\alpha}{2},
			\]
			where the last inequality is because $\alpha > \frac{7}{27} > \frac{22}{87}$. 
			However, according to Claim \ref{clm:diff}, $v(A_1) - v(A_2) \le \alpha - \max_{e\in A_2}v(e) \le \frac{\alpha}{2}$, which is a contradiction. 
			Note that this also implies that except the object with disutility $\alpha$, the total disutility of the other objects must be smaller than $3\cdot \frac{6\alpha-1}{5}$. 
			Since the total disutility of one medium object and one small object is larger than
			\[
			\frac{17\alpha-2}{10} + \frac{6\alpha-1}{5} = \frac{29\alpha-4}{10} \ge \frac{18\alpha-3}{5} = 3\cdot \frac{6\alpha-1}{5},
			\]
			where the inequality is because $\alpha \le \frac{2}{7}$, the only combination that remains to consider is that $A_1$ contains 2 small objects besides the object with disutility $\alpha$. 
			
			\smallskip
			
			\textit{Combination 2}: besides the object with disutility $\alpha$, $A_1$ contains 2 small objects. 
			From the definition of small object, $v(e') \le \max_{e\in A_2}v(e) < \frac{1-\alpha}{5}$ holds for any small object $e' \in A_1$.
			Thus, $v(A_1) < \alpha + 2\cdot \frac{1-\alpha}{5} = \frac{2+3\alpha}{5}$, which is a contradiction to the assumption that $v(A_1) > \Delta_2^{\oplus}(\alpha)$. 
			
			For the other direction, the disutility function for this subcase contains one object with disutility $\alpha$ and five objects with disutility $\frac{1-\alpha}{5}$. 
			Since $\frac{1}{6} < \frac{7}{27} < \alpha \le \frac{2}{7}$, it follows that $\frac{1-\alpha}{5} < \alpha \le 2\cdot \frac{1-\alpha}{5}$. 
			Clearly, the MinMaxShare of this disutility function is $\alpha + 2\cdot \frac{1-\alpha}{5}$. 
			
			\medskip
			
			\textbf{Subcase 3.3: $\alpha \in (\frac{2}{7}, \frac{1}{3}]$} 
			
			\medskip
			
			When $\alpha \in (\frac{2}{7}, \frac{1}{3}]$, $v(A_1) > \Delta_2^{\oplus}(\alpha) = 2\alpha$. 
			If $E_\alpha \cap A_2 \neq \emptyset$, the proof is similar to those for the counterparts of Subcases 3.1 and 3.2. 
			That is, we also have $v(A_1) \le \frac{3-3\alpha}{4}$, which contradicts $v(A_1) > \Delta_2^{\oplus}(\alpha)$ since $\frac{3-3\alpha}{4} < 2\alpha$ when $\alpha > \frac{2}{7} > \frac{3}{11}$. 
			
			Then we focus on $E_\alpha \cap A_2 = \emptyset$. 
			Since $|A_1| \ge 3$, there exists $S \subseteq A_1$ such that $\alpha + \frac{1}{2}(v(A_1)-\alpha) \le v(S) < v(A_1)$. 
			From Claim \ref{clm:v_Aj}, we have a lower bound of $v(A_2)$, i.e., $v(A_2) \ge \alpha + \frac{1}{2}(v(A_1)-\alpha)$. 
			Summing up the lower bounds of $v(A_1)$ and $v(A_2)$ leads to a contradiction,
			\[
			1 = v(A_1) + v(A_2) \ge \frac{3}{2}v(A_1) + \frac{\alpha}{2} > \frac{7\alpha}{2} > 1,
			\]
			where the last inequality is because $\alpha > \frac{2}{7}$. 
			
			For the other direction, 
			the disutility function for this subcase contains three objects with disutility $\alpha$ and one object with disutility $1-3\alpha$ (if $\alpha < \frac{1}{3}$). 
			Since $\alpha > \frac{2}{7} > \frac{1}{4}$, it follows that $1-3\alpha < \alpha$. 
			Clearly, the MinMaxShare is $2\alpha$. 
			
			\subsection{Proof of Theorem \ref{thm:homogeneous}}
			\label{sec:proof:lem:with-m}
			
			We now carefully discuss Hill’s share when $m$ is not sufficiently large, which completes the proof of Theorem~\ref{thm:homogeneous}.
			For the sake of contradiction, we assume that there exists a disutility $v \in \cV(\alpha;m)$ such that $\MMS_n(v) > \Delta_n^{\oplus}(\alpha; m)$,
			and let $\mathbf{A} = (A_1, \ldots, A_n)$ be a lexicographical MinMax allocation of $v$.
			Without loss of generality, assume $v(A_1) \ge \cdots \ge v(A_n)$. We now split the proof into several cases based on the values of $n$ and $k$,
			and 
			it suffices to compute the share for the case 
			where 
			$m$ is smaller than the number of objects in the worst-case disutility function in the unrestricted setting.

			\medskip
			
			\noindent{\bf Case 1: $n \neq 2$ or $k \neq 1$}
			
			\medskip
			
			We consider the subcases $\alpha \in D(n, k)$ and $\alpha \in I(n, k)$, separately. 
			
			\medskip
			
			\textbf{Subcase 1.1: $\alpha \in D(n, k)$ } 
			
			\medskip
			
			Recall that when $\alpha \in D(n, k)$ with $n \neq 2$ or $k \neq 1$, the disutility function constructed in the setting when $m$ is not restricted contains $kn + n + 1$ objects (see Tables \ref{tab:k=0-D} and \ref{tab:general-D}). 
			Therefore, if $m \ge kn + n + 1$, the tight bound remains unchanged.  
			
			Thus we can focus on $m \le kn + n$. 
			Since $v(A_1) > \Delta_n^{\oplus}(\alpha; m) = (k+1)\alpha$, by Claim \ref{clm:v_Aj}, $v(A_j) \ge v(A_1) - \alpha > k\alpha$ for any $j \in N \setminus \{1\}$. 
			Moreover, since the disutility of any object is at most $\alpha$, $A_1$ contains at least $k+2$ objects and $A_j$ contains at least $k+1$ ones, 
			i.e., $|A_1| \ge k+2$ and $|A_j| \ge k+1$. 
			Accordingly, the total number of objects is at least $k+2 + (n-1)(k+1) = kn+n+1 > m$, a contradiction.  
			The disutility function that shows tightness (see Table \ref{tab:restricted:general-D}) contains $\lceil \frac{1}{\alpha} \rceil - 1$ objects with disutility $\alpha$, 
			one object with disutility $ 1 - (\lceil \frac{1}{\alpha} \rceil - 1) \alpha $,
			and $ m - \lceil \frac{1}{\alpha} \rceil$ objects with disutility 0.
			This disutility function is valid since $m \ge \lceil \frac{1}{\alpha} \rceil$. 
			Since $\alpha \in D(n, k)$, $\frac{1}{\alpha} \ge \frac{n(k+1 )^2  + k + 2 }{k+2} \ge kn + 1$, where the last inequality is because $n\ge0$. 
			Therefore, the disutility function contains at least $kn+1$ objects with disutility $\alpha$.
			By the pigeonhole principle, the MinMaxShare is at least $(k+1)\alpha$.

			\begin{table}[htbp]
				\centering
				\begin{tabular}{c|c}
					\hline
					Object Disutility & Quantity \\
					\hline
					$\alpha$ & $\lceil \frac{1}{\alpha} \rceil - 1$ \\
					$ 1 - (\lceil \frac{1}{\alpha} \rceil - 1) \alpha $ & 1 \\
					0 & $ m - \lceil \frac{1}{\alpha} \rceil$ \\
					\hline
				\end{tabular}
				\caption{Disutility function for subcase $\alpha \in D(n, k)$ with $n \neq 2$ or $k \neq 1$, and $m \le kn+n$.}
				\label{tab:restricted:general-D}
			\end{table}
			
			\textbf{Subcase 1.2: $\alpha \in I(n, k)$ } 
			
			\medskip
			
			The bound for $\alpha \in I(n, k)$ remains unchanged regardless of the value of $m$, 
			since there always exists a disutility function whose MinMaxShare is at least $\Delta_n^{\oplus}(\alpha; m) = (k+1)\alpha$. 
			Specifically, the disutility function (see Table \ref{tab:restricted:general-D}) also contains $\lceil \frac{1}{\alpha} \rceil - 1$ objects with disutility $\alpha$,
			one object with disutility $ 1 - (\lceil \frac{1}{\alpha} \rceil - 1) \alpha $, 
			and $ m - \lceil \frac{1}{\alpha} \rceil$ objects with disutility 0.
			Since $\alpha \in I(n,k)$, $\frac{1}{\alpha} \ge kn+1$, which means that there are at least $kn+1$ objects with disutility $\alpha$. 
			By the pigeonhole principle, the MinMaxShare is at least $(k+1)\alpha$. 
			
			\medskip
			
			\noindent{\bf Case 2: $n = 2$ and $k = 1$}
			
			\medskip
			
			Recall that when $n=2$ and $k=1$, $\alpha \in (\frac{1}{5}, \frac{1}{3}]$, thus $m \ge \lceil \frac{1}{\alpha} \rceil \ge 3$.  
			When $m = 3$, $\alpha$ can only be $\frac{1}{3}$. 
			The tight bound remains unchanged (i.e., $\Delta_2^{\oplus}(\frac{1}{3}; 3) = \Delta_2^{\oplus}(\frac{1}{3})$), since the disutility function constructed in the unrestricted setting (i.e., Subcase 3.3 in Subsection \ref{subsubsec:unrestricted:n=2-k=1}) contains 3 objects when $\alpha = \frac{1}{3}$. 
			
			When $m = 4$, $\alpha \in [\frac{1}{4}, \frac{1}{3})$. 
			Since $v(A_1) > \Delta_2^{\oplus}(\alpha; 4) = 2\alpha$, by Claim \ref{clm:v_Aj}, $v(A_2) > \alpha$. 
			Therefore, $A_1$ contains at least 3 objects and $A_2$ contains at least 2 objects, a contradiction to $m = 4$. 
			For the tightness, the disutility function contains $\lceil \frac{1}{\alpha} \rceil - 1$ objects with disutility $\alpha$,
			and one object with disutility $ 1 - (\lceil \frac{1}{\alpha} \rceil - 1) \alpha $. 
			Sine $\frac{1}{\alpha} > 3$, by the pigeonhole principle, the MinMaxShare is at least $2\alpha$. 
			
			When $m = 5$, $\alpha \in (\frac{1}{5}, \frac{1}{3}]$. 
			If $\alpha \in (\frac{1}{5}, \frac{7}{27}]$ or $(\frac{2}{7}, \frac{1}{3}]$, 
			the disutility functions constructed in the unrestricted setting (i.e., Subcases 3.1 and 3.3 in Subsection \ref{subsubsec:unrestricted:n=2-k=1}) contain 5 and 4 objects respectively, thus the tight bounds do not change. 
			If $\alpha \in (\frac{7}{27}, \frac{2}{7}]$, since $v(A_1) > \Delta_2^{\oplus}(\alpha; 5) \ge 2\alpha$, by Claim \ref{clm:v_Aj}, $v(A_2) > \alpha$, thus $A_1$ contains at least 3 objects and $A_2$ contains at least 2 objects. 
			More accurately, since $m = 5$, $|A_1|$ is exactly 3 and $|A_2|$ is exactly 2. 
			Moreover, it can be verified that the largest disutility in $A_1$ is at most the smallest disutility in $A_2$. 
			Since otherwise, by exchanging one object in $A_1$ with a strictly larger disutility and one object in $A_2$ with a strictly smaller disutility, 
			one can get another allocation $\mathbf{A}' = (A'_1, A'_2)$ such that $v(A'_1) < v(A_1)$ and $v(A'_2) \le 2\alpha < v(A_1)$, which contradicts the assumption that $\mathbf{A}$ is a lexicographical MinMax allocation of $v$. 
			Let $A_2 = \{e_1, e_2\}$, it follows that $v(e_1) = \alpha$ and $v(e_2) \ge \frac{1}{3}\cdot v(A_1)$. 
			Therefore, 
			\[
			v(A_1 \cup A_2) \ge v(A_1) + \alpha + \frac{1}{3}v(A_1) = \frac{4}{3}\cdot v(A_1) + \alpha.
			\]
			If $\alpha \in (\frac{7}{27}, \frac{3}{11}]$, $v(A_1) > \Delta_2^{\oplus}(\alpha; 5) = \frac{3-3\alpha}{4}$, thus 
			\[
			v(A_1 \cup A_2) > \frac{4}{3} \cdot \frac{3-3\alpha}{4} + \alpha = 1,
			\]
			a contradiction. 
			If $\alpha \in (\frac{3}{11}, \frac{2}{7}]$, $v(A_1) > \Delta_2^{\oplus}(\alpha; 5) = 2\alpha$, also a contradiction since 
			\[
			v(A_1 \cup A_2) > \frac{11}{3}\alpha > 1. 
			\]
			The disutility function that shows tightness for $\alpha \in (\frac{7}{27}, \frac{3}{11}]$ is the same as that in Subcase 3.1 in Subsection \ref{subsubsec:unrestricted:n=2-k=1}, 
			i.e., one object with disutility $\alpha$ and four objects with disutility $\frac{1-\alpha}{4}$. 
			Again, since $\frac{1}{5} < \frac{7}{27} < \alpha \le \frac{3}{11} < \frac{1}{3}$, $\frac{1-\alpha}{4} < \alpha < 2\cdot \frac{1-\alpha}{4}$, 
			which gives that the MinMaxShare is $\frac{3-3\alpha}{4}$. 
			For $\alpha \in (\frac{3}{11}, \frac{2}{7}]$, the disutility function is the same as that in Subcase 3.3 in Subsection \ref{subsubsec:unrestricted:n=2-k=1}, 
			i.e., three objects with disutility $\alpha$ and one object with disutility $1-3\alpha$. 
			Since $\alpha > \frac{3}{11} > \frac{1}{4}$, $1-3\alpha < \alpha$, thus the MinMaxShare is $2\alpha$. 
			
			When $m \ge 6$, $\alpha \in (\frac{1}{5}, \frac{1}{3}]$. 
			Since the disutility functions constructed in the subcases of the unrestricted setting contain no more than 6 objects, the tight bounds remain unchanged. 
			
			\section{Missing Materials in Section \ref{sec:experiments}}
			
			\subsection{Proof of Claim \ref{claim:r}}
			\label{subsec:proof_claim_r}
			Notice that by Lemma \ref{lem:bestMMS_m} and Lemma \ref{lem:worstMMS_property}, $r_n(\alpha; m)$ is weakly increasing in $m$. 
			Therefore, it suffices to prove the claim for the setting when $m$ is unrestricted, i.e., $r_n(\alpha) \le \frac{2n}{n+1}$. 
			We first consider the case where $n = 2$ and $k = 1$. 
			In this case, $\alpha \in (\frac{1}{5}, \frac{1}{3}]$. 
			Since $\alpha < \frac{1}{n} = \frac{1}{2}$, $\Delta _{2}^{\circleddash}(\alpha) = \frac{1}{2}$. 
			When $\alpha \in (\frac{1}{5}, \frac{7}{27}]$, $\Delta_2^{\oplus}(\alpha) = \frac{3 - 3\alpha}{4}$, thus $r_2(\alpha) = \frac{3 - 3\alpha}{2} < \frac{6}{5} < \frac{4}{3}$; 
			when $\alpha \in (\frac{7}{27}, \frac{2}{7}]$, $\Delta_2^{\oplus}(\alpha) = \frac{2 + 3\alpha}{5}$, thus $r_2(\alpha) = \frac{4 + 6\alpha}{5} \le \frac{8}{7} < \frac{4}{3}$; 
			when $\alpha \in ( \frac{2}{7}, \frac{1}{3}]$, $\Delta_2^{\oplus}(\alpha) = 2\alpha$, thus $r_2(\alpha) = 4\alpha \le \frac{4}{3}$. 
			
			We next consider the cases when $n \ge 3$ or $k \neq 1$. 
			When $\alpha > \frac{1}{n}$ which means $\alpha \in I(n, 0)$ or $\alpha \in (\frac{1}{n}, \frac{2}{n+2}] \in D(n, 0)$, $\Delta_n^{\circleddash}(\alpha) = \alpha$. 
			Thus, when $\alpha \in I(n, 0)$, $\Delta_n^{\oplus}(\alpha) = \alpha$ and $r_n(\alpha) = 1 < \frac{4}{3}\le \frac{2n}{n+1}$ since $n \ge 2$; 
			when $\alpha \in (\frac{1}{n}, \frac{2}{n+2}]$, $\Delta_n^{\oplus}(\alpha) = \frac{2\cdot(1-\alpha)}{n}$ and $r_n(\alpha) = \frac{2}{n}\cdot \frac{1-\alpha}{\alpha} < 2\cdot (1-\frac{1}{n}) < \frac{2n}{n+1}$. 
			When $\alpha \le \frac{1}{n}$, it follows that $\alpha \in (\frac{1}{n+1}, \frac{1}{n}] \in D(n, 0)$ or $\alpha \in I(n, k)$ with $k \ge 1$ or $\alpha \in D(n, k)$ with $k \ge 1$. 
			In these cases, $\Delta_n^{\circleddash}(\alpha) = \frac{1}{n}$. 
			When $\alpha \in (\frac{1}{n+1}, \frac{1}{n}]$, $\Delta_n^{\oplus}(\alpha) = \frac{2\cdot(1-\alpha)}{n}$ and $r_n(\alpha) = 2\cdot(1-\alpha) < \frac{2n}{n+1}$; 
			when $\alpha \in I(n, k) = (\frac{k+2}{n(k+1 )^2  + k + 2 }, \frac{1}{kn+1}]$ with $k \ge 1$, $\Delta_n^{\oplus}(\alpha) = (k+1)\alpha$ and $r_n(\alpha) = n(k+1)\cdot \alpha \le \frac{kn+n}{kn+1} \le \frac{2n}{n+1}$; 
			when $\alpha \in D(n, k) = (\frac{1}{kn+n+1}, \frac{k+2}{n(k+1 )^2  + k + 2 }]$ with $k \ge 1$, $\Delta_n^{\oplus}(\alpha) = \frac{k+2}{k+1}\cdot\frac{1-\alpha}{n}$ and $r_n(\alpha) = \frac{k+2}{k+1}\cdot (1-\alpha) < \frac{kn+2n}{kn+n+1} \le \frac{3n}{2n+1} < \frac{2n}{n+1}$. 
			
			\subsection{Proof of Claim \ref{claim:large}}
			Note that we actually derive the ranges of $\alpha$ that satisfy $r_n(\alpha; +\infty) > \frac{4}{3}$, which are necessary conditions for $r_n(\alpha; m) > \frac{4}{3}$ but may not be sufficient ones. 
			We use the formulas of $r_n(\alpha)$ derived in the proof of Claim \ref{claim:r} in Subsection \ref{subsec:proof_claim_r}, and only consider the following cases when $r_n(\alpha)$ may be larger than $\frac{4}{3}$. 
			\begin{itemize}
				\item When $\alpha \in (\frac{1}{n}, \frac{2}{n+2}]$, $r_n(\alpha) = \frac{2}{n}\cdot \frac{1-\alpha}{\alpha}$, which is larger than $\frac{4}{3}$ when $\alpha < \frac{3}{2n+3}$. 
				Since $\frac{3}{2n+3} > \frac{1}{n}$ only when $n \ge 4$ and $\frac{3}{2n+3} < \frac{2}{n+2}$, the range is $\alpha \in (\frac{1}{n}, \frac{3}{2n+3})$ with $n \ge 4$.
				\item When $\alpha \in (\frac{1}{n+1}, \frac{1}{n}]$, $r_n(\alpha) = 2\cdot(1-\alpha)$, which is larger than $\frac{4}{3}$ when $\alpha < \frac{1}{3}$. 
				Since $\frac{1}{n+1} < \frac{1}{3}$ only when $n\ge 3$ and $\frac{1}{n} \le \frac{1}{3}$ when $n \ge 3$, the range is $\alpha \in (\frac{1}{n+1}, \frac{1}{n})$ with $n \ge 3$. 
				\item When $\alpha \in I(n, k) = (\frac{k+2}{n(k+1 )^2  + k + 2 }, \frac{1}{kn+1}]$ with $k \ge 1$, $r_n(\alpha) = n(k+1)\cdot \alpha$, which is larger than $\frac{4}{3}$ when $\alpha > \frac{4}{3n(k+1)}$. 
				Note that $\frac{4}{3n(k+1)} < \frac{1}{kn+1}$ is equivalent to $(3-k)n > 4$, which can be satisfied only when $k = 1$ or $k = 2$. 
				When $k = 1$, $(3-k)n > 4$ gives $n \ge 3$, $\alpha > \frac{4}{3n(k+1)}$ is equivalent to $\alpha > \frac{2}{3n}$, and $\frac{k+2}{n(k+2)^2+k+2} = \frac{3}{4n+3}$. 
				Since $\frac{3}{4n+3} \ge \frac{2}{3n}$ when $n \ge 6$, the ranges are $\alpha \in (\frac{2}{3n}, \frac{1}{n+1})$ with $3 \le n \le 5$, and $\alpha \in (\frac{3}{4n+3}, \frac{1}{n+1})$ with $n \ge 6$. 
				When $k = 2$, $(3-k)n > 4$ gives $n \ge 5$, $\alpha > \frac{4}{3n(k+1)}$ is equivalent to $\alpha > \frac{4}{9n}$, and $\frac{k+2}{n(k+2)^2+k+2} = \frac{1}{4n+1}$. 
				Since $\frac{4}{9n} > \frac{1}{4n+1}$, the range is $\alpha\in (\frac{4}{9n}, \frac{1}{2n+1})$ with $n \ge 5$.
				\item When $\alpha \in D(n, k) =(\frac{1}{kn+n+1}, \frac{k+2}{n(k+1 )^2  + k + 2 }]$ with $k \ge 1$, $r_n(\alpha) = \frac{k+2}{k+1}\cdot (1-\alpha)$, which is larger than $\frac{4}{3}$ when $\alpha < \frac{2-k}{3k+6}$. 
				Note that $\frac{2-k}{3k+6} > 0$ only when $k = 1$. 
				Then, $\alpha \le \frac{2-k}{3k+6}$ is equivalent to $\alpha < \frac{1}{9}$, $\frac{1}{kn+n+1}=\frac{1}{2n+1}$ and $\frac{k+2}{n(k+1)^2+k+2} = \frac{3}{4n+3}$. 
				Since $\frac{3}{4n+3} \le \frac{1}{9}$ when $n \ge 6$ and $\frac{1}{9} > \frac{1}{2n+1}$ when $n \ge 5$, the ranges are $(\frac{1}{2n+1}, \frac{3}{4n+3})$ with $n \ge 6$, and $(\frac{1}{2n+1}, \frac{1}{9})$ with $n = 5$. 
			\end{itemize}
			By summarising the above ranges, we complete the proof.

			\subsection{More Experiments}
			
			We observe that in Fig. \ref{fig:ratio:random}, when $n=2$, the majority of random instances fall into the interval of $[1.1, 1.2)$, in contrast to the other values of $n$ that are concentrated within $[1.0, 1.1)$.
			This is in part because the ratio of $m$ over $n$ is larger than $n>2$, given each $m$.
			One may be curious that when $m$ becomes larger and larger to $n$, the majority may be close to the worst-case ratio. 
			Due to this curiosity, we further conduct the following experiment by setting $m = 15 \pm 1$ and $m= 20 \pm 1$, where $n$ is fixed at 2.
			The results are shown in Fig. \ref{fig:n=2:m=15-20}.
			As we can see, the instances get more concentrated within $[1.1, 1.2)$, and the number of instances whose ratios are above 1.2 get less and less.
			
			\begin{figure}[H]
				\centering
				\subfigure[$n=2$, $m=8, 9, 10$]{
					\includegraphics[width=0.30\linewidth]{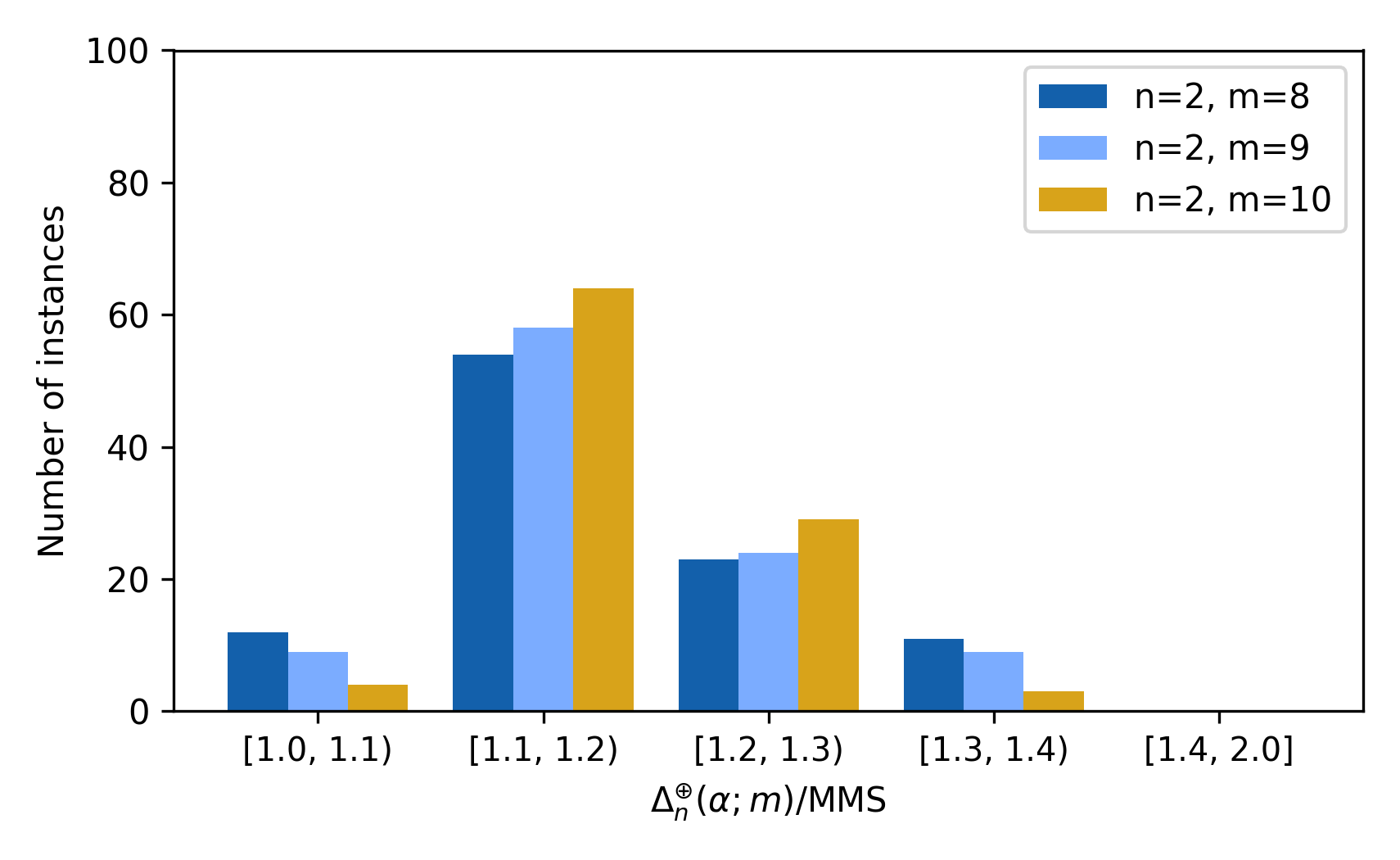}
				}
				\,
				\subfigure[$n=2$, $m=14, 15, 16$]{
					\includegraphics[width=0.30\linewidth]{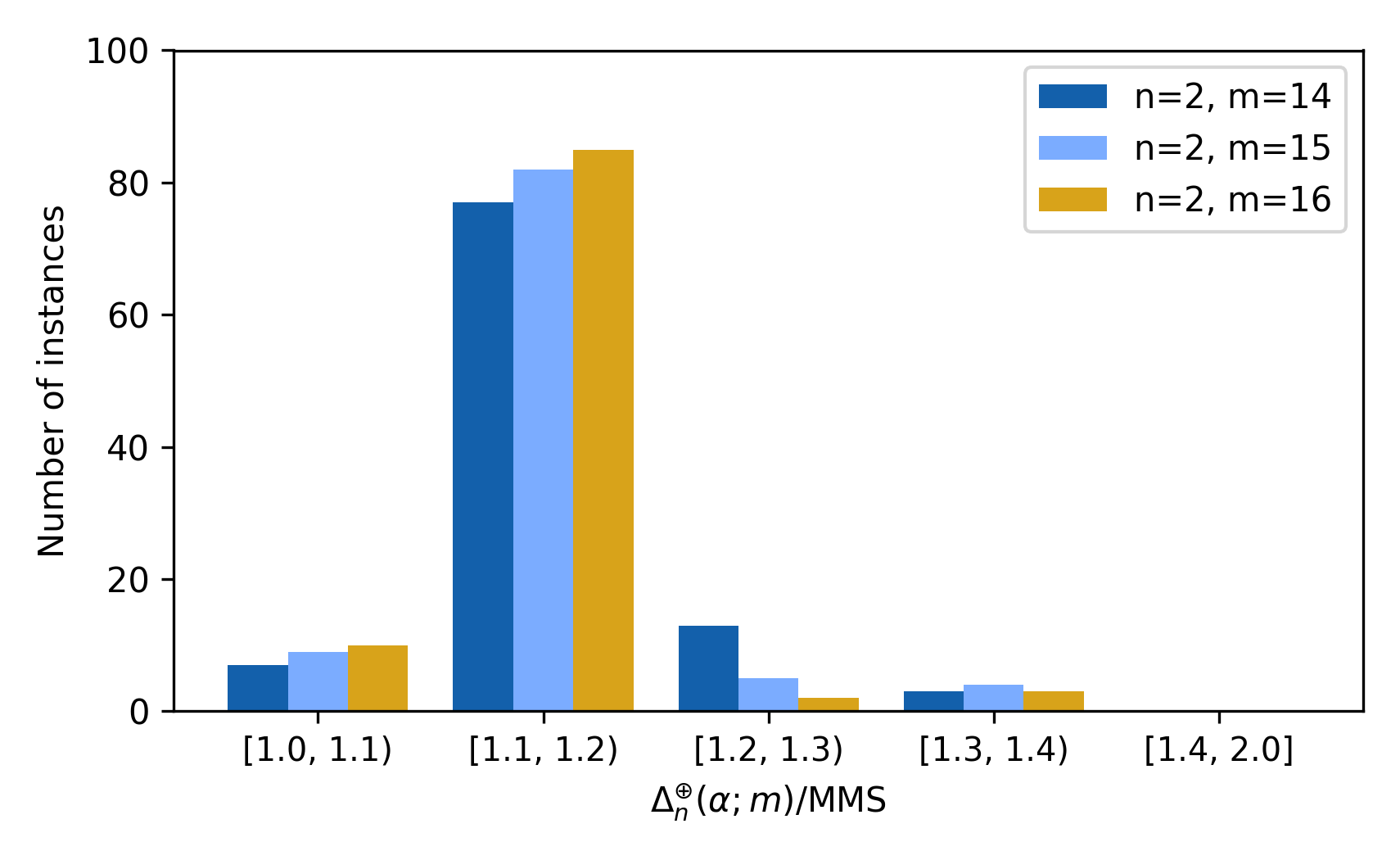}
				}
				\,
				\subfigure[$n=2$, $m=19, 20, 21$]{
					\includegraphics[width=0.30\linewidth]{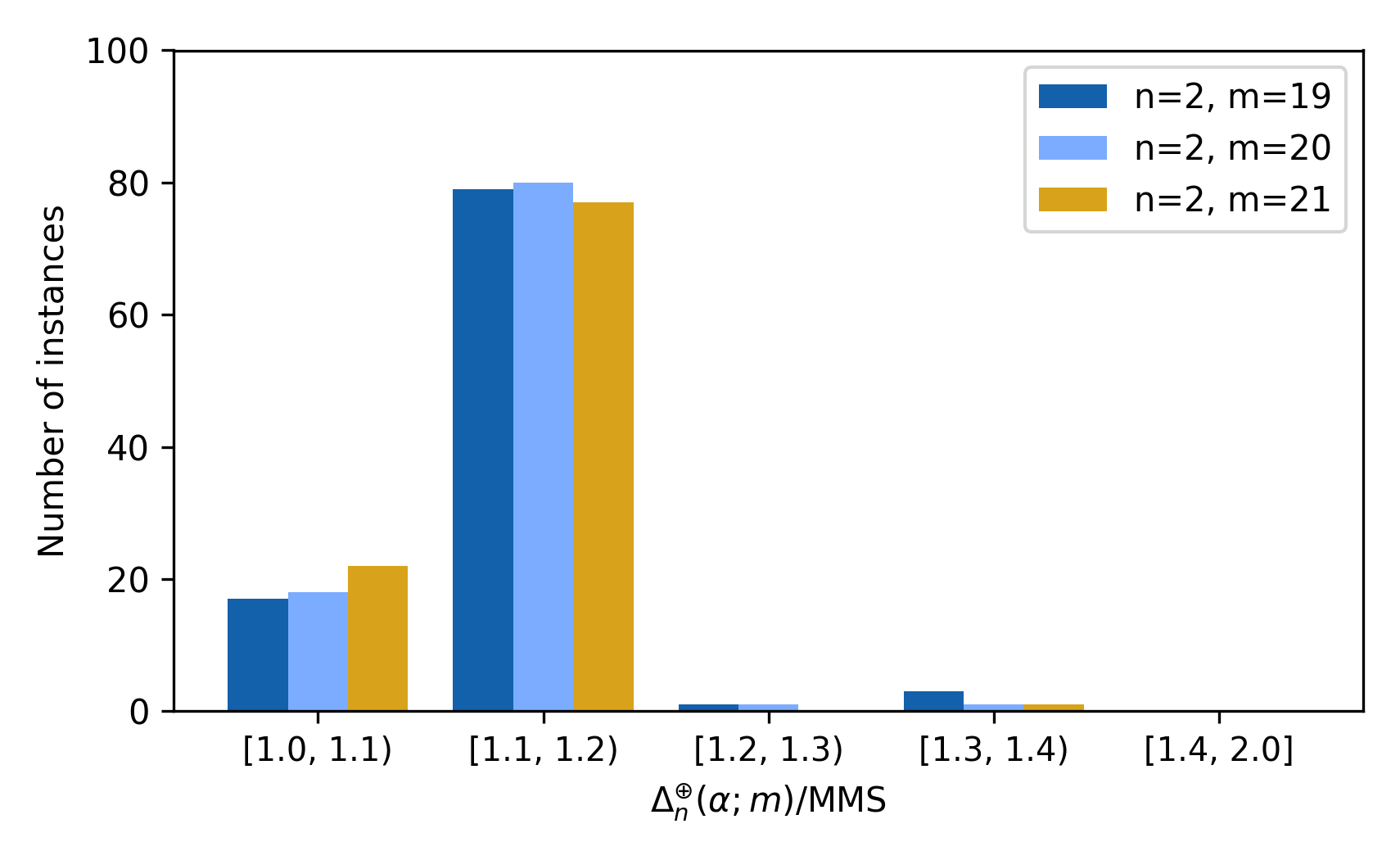}
				}
				\caption{Fixing $n = 2$ and increasing the value of $m$.}
				\label{fig:n=2:m=15-20}
			\end{figure} 
			

	\end{document}